\documentclass[sigconf]{acmart}

\acmConference[WSDM '25]{The 18th ACM International Conference on Web Search and Data Mining}{March 10th-14th, 2025}{Hannover, Germany}
\acmBooktitle{Proceedings of the ACM WSDM Conference}
\acmPrice{}
\acmISBN{978-1-4503-XXXX-X/XX/XX}
\acmDOI{10.1145/XXXXXXX.XXXXXXX}
\usepackage{graphicx}  % Required for including images
\usepackage{float}     % Required for placing figures 
\usepackage{amsfonts}  % Provides AMS fonts
\usepackage{amssymb,amsmath,amsthm}  % AMS symbols and mathematics
\usepackage{tikz}      % TikZ package for drawing
\usepackage{xcolor}    % Color package for colored text
\usepackage{booktabs}
\usepackage{mdframed}
\newmdenv[
  linewidth=1pt,
  roundcorner=5pt,
  backgroundcolor=blue!10,
  linecolor=black,
  innertopmargin=5pt,
  innerbottommargin=5pt,
  innerleftmargin=5pt,
  innerrightmargin=5pt
]{problem}
\usetikzlibrary{positioning}  % Additional TikZ library for positioning
\DeclareMathOperator{\diag}{diag} % Diagonal operator
\begin{document}
\title{Do Stubborn Users Always Cause More Polarization and Disagreement? A Mathematical Study}

\author{Mohammad Shirzadi}
\email{mohammad.shirzadi@anu.edu.au}
\affiliation{
  \institution{The Australian National University}
  \city{Canberra}
  \country{Australia}
}

\author{Ahad N. Zehmakan}
\email{ahadn.zehmakan@anu.edu.au}
\affiliation{
  \institution{The Australian National University}
  \city{Canberra}
  \country{Australia}
}

\begin{abstract}
We study how the stubbornness of social network users influences opinion polarization and disagreement. Our work is in the context of the popular Friedkin-Johnson opinion formation model, where users update their opinion as a function of the opinion of their connections and their own innate opinion. Stubbornness then is formulated in terms of the stress a user puts on its innate opinion.

We examine two scenarios: one where all nodes have uniform stubbornness levels (homogeneous) and another where stubbornness varies among nodes (inhomogeneous). In the homogeneous scenario, we prove that as the network's stubbornness factor increases, the polarization and disagreement index grows.
In the more general inhomogeneous scenario, our findings surprisingly demonstrate that increasing the stubbornness of some users (particularly, neutral/unbiased users) can reduce the polarization and disagreement. We characterize specific conditions under which this phenomenon occurs. Finally, we conduct an extensive set of experiments on real-world network data to corroborate and complement our theoretical findings.
\end{abstract}


\maketitle

\section{Introduction}
Social media experienced significant growth in the early 2000s with the advent of several pivotal platforms. This rise of large social media networks has significantly advanced global information sharing and human interaction. While the enhanced accessibility of information and heightened online engagement are positive aspects of this social transformation, there is mounting concern regarding the potential of these platforms to exacerbate societal polarization and division. For instance, studies by Hart and Nisbet~\cite{hart2012boomerang} and DellaPosta et al.~\cite{dellaposta2020pluralistic} highlight the severity of this issue within the contexts of climate change and political discourse.

There has been a fast-growing multidisciplinary research effort, cf.~\cite{huang2022pole,haddadan2021repbublik,chitra2020analyzing,musco2018minimizing}, to understand how individuals form their opinions through social interactions and what is the root cause of phenomena such as polarization and disagreement in social networks. From a mathematical perspective, it is natural to introduce and study models which simulate the opinion formation process.

One opinion formation model which has garnered a significant amount of attention, cf.~\cite{friedkin1990social,out2024impact}, is the Friedkin-Johnson (FJ) model, where each user has an opinion between $-1$ and $+1$. Then, over time, each user updates its opinion as a function of the opinion of its neighbors and its own innate opinion. The FJ model owns its popularity mostly to the facilitation of a rich framework to mathematically formulate and investigate concepts such as polarization and agreement, cf.~\cite{musco2018minimizing, zhu2021minimizing,racz2023towards}. Polarization is defined as the variance in opinions of the users, while the disagreement between two connected users is measured by the difference in their opinions.

The formation of polarization and disagreement in society appears to be influenced by intrinsic human social dynamics, as well as the structural and functional characteristics inherent to social media platforms. One example of human behavior that affects polarization and disagreement in social networks is stubbornness. Stubbornness measures how much an individual adheres to their innate opinion or belief versus their willingness to change their opinions in response to their neighbors' views, cf.~\cite{abebe2018opinion}.

While the notion of stubbornness has been introduced and studied in various setups~\cite{steinberg2007age,wall1993susceptibility}, particularly in the FJ model~\cite{xu2022effects}, our understanding of how the stubbornness of users contributes to the formation of polarized communities or disagreement is very limited. The present paper aims to provide novel and deep insights on the relation between stubbornness parameter and polarization-disagreement ($\mathcal{PD}$) index. More precisely, we provide a collection of theoretical results, accompanied by experiments on real-world graph data, which characterize the settings where an increase in stubbornness will result in an increase or, surprisingly, decrease in polarization-disagreement. 

\textbf{Roadmap.} In the rest of this section, we provide preliminaries, give a short overview of our findings, and review related prior work. In Section~\ref{MAPD_section}, some theoretical foundations are established. In Section~\ref{homogeneous_section}, we analyze the homogeneous stubbornness case and prove that any increase in the stubbornness level amplifies the $\mathcal{PD}$ index. In Section~\ref{inhomogeneous_section}, we delve into inhomogeneous cases and prove that if the stubbornness of individuals with a neutral (zero) innate opinion increases, the $\mathcal{PD}$ will decrease. Experiments on some synthetic and real-world data sets are presented in Section~\ref{experiments}.
Finally, Section~\ref{conclusion} concludes the paper with a brief summary. 

\vspace{-0.2cm}

\subsection{Preliminaries}
\subsubsection{Basic Definitions}
Let $G=(V,E,w)$ be an undirected, weighted, and connected graph, where $V=\{v_1,\cdots,v_n\}$, $E\subseteq \{\{v_i,v_j\}: 1\le i,j\le n\}$, and $w: E \rightarrow \mathbb{R}^+$. We define $m=\sum_{\{u,v\} \in E} w_{uv}$ to be the sum of the weights of all edges (in the unweighted case where all weights are one, this sum is just the total numbers of edges). For each node $v_i$, $N(v_i)=\{u\in V|\{u,v_i\} \in E\}$ denotes the neighborhood of node $i$. The graph $G$ represents a social network, where each nod correspond to a user. An edge between two nodes indicates that the corresponding users are connected, and the weight of the edge shows the strength of their relation.

We define $\mathbf{D}$ to be the diagonal degree matrix given by $D_{ii}=\sum_{v_j \in N(v_i)} w_{v_iv_j}$, $\textbf{A} \in \mathbb{R}^{n \times n}$ to be the weighted adjacency matrix and $\textbf{L}:=\textbf{D}-\textbf{A}$ denotes the Laplacian matrix. For any connected graph $G$ as above, the Laplacian matrix is positive semi-definite, i.e., $\textbf{L} \succeq 0$ (for more details see, e.g.,~\cite{golub2013matrix}) and it can be written as 
\begin{equation*}
    \textbf{L}=\sum_{j=1}^n \lambda_j q_j q_j^T,
\end{equation*}
where $0=\lambda_1< \lambda_2 \leq \cdots \leq \lambda_n$ are ordered eigenvalues and $q_j$, for $j=1,\cdots,n$, are the corresponding orthonormal eigenvectors.

Given a symmetric matrix $\textbf{L} \in \mathbb{R}^{n \times n}$ with described eigendecomposition as above and function $f:\mathbb{R} \rightarrow \mathbb{R}$, define 
\begin{equation} \label{f(L)}
    f(\textbf{L})=\sum_{j=1}^n f(\lambda_j) q_j q_j^T,
\end{equation}
where we suppose that $f(y) \geq 0$ for $y \geq 0$, which yields $f(\textbf{L}) \succeq 0$ for $\textbf{L} \succeq 0$.

For $p \geq 1$ real number, the $l^p$-norm of vector $x=[x_1,x_2,\cdots,x_n]^T$ is defined as 
\footnote{
For $p=1$, we have the Taxicab norm or Manhattan norm, for $p=2$ we will have the Euclidean norm, and if p approaches infinity, we get the maximum norm $\|x\|_{\infty}:=\max_{i} |x_i|$.
}
\begin{equation*}
   \|x\|_p:=(\sum_{i=1}^n |x_i|^p )^{1/p}.
\end{equation*}
For simplicity, in the case of $p=2$, we use $\|x \|_2:=\| x\|$.  The norms $\|\textbf{A}\|_1$ and $\| \textbf{A}\|_{\infty}$ of a matrix $\textbf{A} \in \mathbb{R}^{n \times n}$ are defined as 
\begin{equation*}
    \|\textbf{A}\|_1=\max_{j} \sum_{i=1}^n |a_{ij}|, \quad     \|\textbf{A}\|_{\infty}=\max_{i} \sum_{j=1}^n |a_{ij}|. 
\end{equation*}

The Courant-Fischer theorem, cf.~\cite{horn2012matrix}, provides a variational formulation for the eigenvalues of a symmetric matrix, which can help derive bounds on the eigenvalues. In general, if $\textbf{A}$ is an $n \times  n$ symmetric matrix with eigenvalues $\lambda_1 \leq \lambda_2 \leq \cdots \leq \lambda_n$ with corresponding eigenvector $v_1,\cdots,v_n$, then
\begin{equation*}
    \lambda_1=\min_{x \neq 0} \frac{x^T \textbf{A} x}{x^T x}, \quad 
    \lambda_2=\min_{x \neq 0, x \perp v_1} \frac{x^T \textbf{A} x}{x^T x}, 
    \quad 
    \lambda_{\max}=\max_{x \neq 0} \frac{x^T \textbf{A} x}{x^T x}. 
\end{equation*}
Finally, if $a_i$ and $b_i$ are two distinct sequences, we will have 
\begin{equation}\label{multiplication_series}
\bigg(\sum_{i=1}^n a_i\bigg) \bigg(\sum_{i=1}^n b_i\bigg)= \sum_{i=1}^n a_i b_i+\sum_{i=1}^n \sum_{j \neq i} a_i b_j.
\end{equation}
\subsubsection{Friedkin-Johnson Dynamics}

The FJ model is a foundational framework for understanding opinion formation as described by Friedkin and Johnsen~\cite{friedkin1990social}. In this model, each node \( v_i \in V \) possesses an innate opinion \( s_i: v_i \rightarrow [-1,1] \), which is assumed to remain constant over time. This innate opinion represents an individual's perspective in a hypothetical social vacuum, free from external influences, and is typically shaped by factors such as background, location, religion, and other personal circumstances. Additionally, each node \( v_i \) is characterized by a \textit{stubbornness coefficient}, denoted as \( k_i > 0 \), which indicates the weight an individual assigns to their own opinion. A higher \( k_i \) value suggests a more stubborn individual \( v_i\), whereas a lower \( k_i \) value implies a more open-minded individual \( v_i \). Beyond their innate opinion, each node \( v_i\), also has an expressed opinion \( z_i^{(t)} \in [-1,1] \) at any given time \( t \). This expressed opinion evolves through a process of averaging, influenced by both the node's innate opinion (modulated by their stubbornness) and the expressed opinions of their neighbors. More precisely, the expressed opinion of node $v_i$ in time step $t$ is formulated as follows:

\begin{equation}\label{FJ}
    z_i^{(t)}=\frac{k_i s_i+\sum_{j \in N(v_i)} w_{v_iv_j} z_j^{(t-1)}}{k_i+\sum_{j \in N(v_i)} w_{v_iv_j}}. 
\end{equation}
It should be noted that in the classic FJ model, all \( k_i \) values are set to $1$. We refer to this special case as classic or vanilla FJ model in the rest of the paper.

The following lemma demonstrates that the equilibrium opinion, similar to the classic FJ model, does not depend on the initial expressed opinions but rather on the stubbornness matrix $\textbf{K}=\diag(k_1,k_2,\cdots,k_n)$ and the innate opinion vector \( s = [s_1, \cdots, s_n]^T \).

\begin{lemma} \label{equlibrium}
The iterative averaging process defined by the update rule in Equation \eqref{FJ} converges to the Nash equilibrium \( z^* = \big( \mathbf{L} + \mathbf{K} \big)^{-1} \mathbf{K} s \).
\end{lemma}

\begin{proof}
    To determine the equilibrium opinion, we follow the approach outlined in~\cite{bindel2015bad}. The updating rule in Equation~\eqref{FJ} can be reformulated as
%\begin{equation}\label{update_FJ}
    $z^{(t)}=(\textbf{K}+\textbf{D})^{-1}\big(\textbf{K}\mathbf{s}+\textbf{A}z^{(t-1)}\big)$
%\end{equation}
where $\textbf{K}=\diag(k_1,k_2,\cdots,k_n)$ denotes the stubbornness matrix.  If $z^{(t)}$ and $y^{(t)}$ are arbitrary vectors, then
%\ahad{Both equations below need to be broken down to more than one line. Now, they go to the margin.\mohammad{red}{This is fixed}}
\begin{equation*}
    \| z^{(t)}-y^{(t)}\|_{\infty} \leq \frac{\big\| \textbf{\textbf{A}}(z^{(t-1)}-y^{(t-1)})\big\|_{\infty}}{\| \textbf{K}+\textbf{D} \|_{\infty}}
\end{equation*}
\begin{equation*}
\quad \quad \quad \quad \quad \quad \quad \quad \quad \quad \leq \frac{\| \textbf{A} \|_{\infty}}{\| \textbf{K}+\textbf{D} \|_{\infty}} \big\| z^{(t-1)}-y^{(t-1)}\big\|_{\infty}, 
\end{equation*}
where by definition we have
\begin{equation*}
   \frac{\| \textbf{A} \|_{\infty}}{\| \textbf{K}+\textbf{D} \|_{\infty}} =  \frac{\max_{i} |\sum_{j} \textbf{A}_{ij}|}{\max_{i} |\sum_{j} (\textbf{K}_{ij}+\textbf{D}_{ij})|}
\end{equation*}
\begin{equation*}
\quad \quad \quad \quad \quad \quad \quad \quad \quad \quad  =\frac{\max_{i} \sum_{j \in N(v_i)} w_{v_iv_j}}{\max_{i} k_i+\max_{i} \sum_{j \in N(v_i)} w_{v_iv_j}} <1.
\end{equation*}
%\begin{equation*}
%\quad \quad \quad \quad \quad  < 1, 
%\end{equation*}
This means that $z^{(t)}$ is a contraction mapping in the max norm. Hence, the iteration update rule by FJ model converges to a unique fixed point, which can be obtained by taking the limit on both sides of the iterative formula and utilizing the definition of the Laplacian matrix as 
%\begin{equation} \label{z*}
    $z^*=\big(\textbf{L}+\textbf{K}\big)^{-1} \textbf{K}\mathbf{s}$.
%\end{equation}
\end{proof}

\subsubsection{Polarization and Disagreement.} %Note that the eigenvalues of matrix $\textbf{K}$ are $\{k_i\}_{i=1}^n$ and the corresponding eigenvectors are given by column matrices $\textbf{e}_j$ where the $j$-th entry is $1$ and all other entries are zero. 
To define polarization and disagreement, we need to define the mean-centered equilibrium vector \(\bar{z}\) as
\begin{equation}
    \bar{z}:=z^*-\frac{{z^*}^T \Vec{\textbf{1}}} {{n}} \Vec{\textbf{1}}, 
\end{equation}
where $\Vec{\textbf{1}}$ is the column vector of size $n$ and all entries are one. 

\begin{definition}[Disagreement]
Disagreement is defined as 
\begin{equation*}
    \mathcal{D}_{G,z^*}:=\sum_{\{u,v\} \in E} w_{uv} (z^*_u-z^*_v)^2= {z^*}^T \textbf{L} z^*. 
\end{equation*}
\end{definition}
% Our primary interest lies in the disagreement at the final configuration. Therefore, \(\mathcal{D}_{G,z^*}\) represents the squared differences between the opinions of all connected nodes at equilibrium.
As $(z_u^*-z_v^*)^2=(z_u^*-\frac{{z^*}^T \Vec{\textbf{1}}}{n}-z_v^*+\frac{{z^*}^T \Vec{\textbf{1}}}{n})^2=(\bar{z}_u-\bar{z}_v)^2$, and $\mathcal{D}_{G,z^*}={z^*}^T \textbf{L} z^*$, so $\mathcal{D}_{G,z^*}=\bar{z}^T \textbf{L} \bar{z}$, cf.~\cite{musco2018minimizing}. 

\begin{definition}[Polarization]
Polarization is quantified by the variance of opinions from the average and given by
\begin{equation*}
    \mathcal{P}_{G,z^*}:=\sum_{u \in V} (z_u^*-\frac{(z^*)^T \Vec{\textbf{1}}}{n} \Vec{\textbf{1}})^2=\bar{\textbf{z}}^T \bar{\textbf{z}}.%=\sum_{u \in V} \bar{z}_u^2=\bar{\textbf{z}}^T \bar{\textbf{z}}. 
\end{equation*}
\end{definition}
% Note that polarization at the final configuration computed as $\mathcal{P}_{G,z^*}=\bar{\textbf{z}}^T \bar{\textbf{z}}$.
Intuitively, polarization measures how opinions at equilibrium deviate from the average.

It is worth mentioning that both indices can be defined for any expressed opinion vector $z$, but we are interested in the equilibrium expressed opinion $z^*$.

\begin{definition}[Polarization-Disagreement]
The Polarization- Disagreement 
($\mathcal{PD}$) index is defined simply as the sum of polarization and disagreement:
\begin{equation*}
\mathcal{PD}_{G,z^*}:=\mathcal{P}_{G,z^*}+\mathcal{D}_{G,z^*}. 
\end{equation*}
\end{definition}

$\mathcal{PD}$ index received special attention by prior work since it has been observed that polarization and disagreement usually don't go hand in hand, and reducing one causes increasing the other. Thus, one ideally aims to reduce their sum, cf.~\cite{musco2018minimizing}. As an example, within social networks, connections between users who share similar mindsets are preferred over those between individuals with differing viewpoints. Thus, recommended links aim to decrease the disagreement experienced by users. However, this effort to minimize disagreement can result in increased polarization.

\textbf{Note.} In this study, we adhere to more commonly used definitions of $\mathcal{P}$, $\mathcal{D}$, and $\mathcal{PD}$, cf.~\cite{musco2018minimizing}. However, Xu et al.~\cite{xu2022effects} proposed an alternative definition. Our proof techniques can be adjusted to generate analogous results for this version of definitions. More details are provided in the Appendix \ref{Xue_result_appendix}.

\vspace{-0.1cm}
\subsection{Our contribution}

\textit{When does increasing the stubbornness of a node causes an increase or decrease in $\mathcal{PD}$ index?} Our main goal is to address this question.

As a preliminary step, we start by providing some matrix formulations for $\mathcal{PD}$ index and upper-bounds in terms of the stubbornness matrix and eigenvalues of a matrix, depending on graph structure and stubbornness values. 

Our first main result is in the homogeneous setup, where all nodes have the same stubbornness $\alpha$. We prove that increasing stubbornness $\alpha$ causes an increase in $\mathcal{PD}$. As a side product, we also provide similar results for polarization index and give an upper bound on the increase in polarization in terms of the stubbornness values, pre- and post-increase.

To gain a deeper understanding of how nodes' stubbornness impacts the opinion formation in social networks, we examine the stochastic block model (a standard generative model for social networks). We derive a closed-form formula for the $\mathcal{PD}$ in a simplified setup. We then demonstrate that when the number of nodes in the graph is sufficiently large, the $\mathcal{PD}$ index is \textit{quadratically} affected by the value of stubbornness, in the homogeneous setup.

The observed positive correlation between network's stubbornness and $\mathcal{PD}$ index is somewhat intuitive. Individuals with high levels of stubbornness are less likely to reconcile with those holding opposing viewpoints, causing an increase in the $\mathcal{PD}$. Conversely, more open-minded individuals are more likely to adjust their opinion to be aligned with others, reducing the $\mathcal{PD}$.

The above results raise the question: \textit{Can an increase in a node's stubbornness result in a reduction in the $\mathcal{PD}$?} To answer this question, we study the more general inhomogeneous setup, where nodes can have different values of stubbornness. We answer this question affirmatively by proving that if we increase the stubbornness of a \textit{neutral} node (a node whose opinion is equal to the average opinion $0$), the $\mathcal{PD}$ index decreases. We, in fact, characterize the exact change in the $\mathcal{PD}$.

In addition to our theoretical findings, we also conduct a set of experiments on both real-world and synthetic graph data. We first observe, as one might expect, that in most randomly generated setups, increasing the stubbornness of nodes increases the $\mathcal{PD}$. However, if we increase the stubbornness of a neutral node, the $\mathcal{PD}$ decreases, in line with our theoretical findings. We also observe that, interestingly, the degree of a node doesn't seem to play a role on the magnitude of the change (increase/decrease) in the $\mathcal{PD}$.

\textit{Are neutral nodes the only ones whose increase in stubbornness can reduce the $\mathcal{PD}$?} We observe that, surprisingly, sometimes making nodes with extreme opinions more stubborn could help reduce the $\mathcal{PD}$. This, in particular, happens when we have a well-connected community of nodes who are biased towards an opinion (mimicking a social bubble in real-world examples), and then we make a node with an opposing opinion in that bubble more stubborn.

\subsection{Related Work}

\textbf{Opinion Formation Dynamics.}  
Several opinion formation models have been introduced due to their various applications, such as understanding political voting dynamics~\cite{acemoglu2011opinion,zehmakan2019spread} and viral marketing strategies~\cite{zehmakan2024viral,zehmakan2021majority}.
DeGroot~\cite{degroot1974reaching} pioneered an opinion dynamics model in his seminal work on consensus formation, which suggested that individual opinions evolve through averaging opinions within their social network neighborhood. Friedkin and Johnsen~\cite{friedkin1990social} expanded on the DeGroot model to incorporate both disagreement and consensus, integrating each individual's initial beliefs into the averaging process. Our work is in the framework of this model.

Some other well-studied opinion dynamics include majority-based opinion formation models~\cite{zehmakan2023random,gartner2018majority}, voter models~\cite{das2014modeling}, the Galam model~\cite{gartner2020threshold,gartner2020threshold}, and bounded confidence models~\cite{deffuant2000mixing}. In a comprehensive survey of opinion diffusion models, Noorazar~\cite{noorazar2020recent} provides a detailed examination of these and other models.

\textbf{Polarization and Disagreement.} 
There has been substantial recent work on employing opinion dynamics models to study polarization and disagreement in social networks, cf.~\cite{rastegarpanah2019fighting,haddadan2021repbublik}. Here, we delve into research centered on the FJ model. Musco et al.~\cite{musco2018minimizing} were pioneers in defining the $\mathcal{PD}$ index within this model, establishing its convexity and demonstrating its solvability in polynomial time when minimizing the index with a specified budget over innate opinion. Zhu et al.~\cite{zhu2021minimizing} further extended this work by proving that adding an edge in a given graph decreases the $\mathcal{PD}$. They also developed a cubic-time greedy algorithm for recommending a specified number of new links to minimize $\mathcal{PD}$ in social networks. In a similar theoretical framework, Wang et al.~\cite{wang2024relationship} explored identifying social links that most effectively reduce $\mathcal{PD}$ using diverse graph machine learning algorithms. Learning-based algorithms are presented in \cite{sun2023self, sun2023all} from an empirical analysis perspective.

Rácz et al.~\cite{racz2023towards} introduced several novel proof techniques showing that the polarization objective function is non-convex, challenging the suggestion that adding edges necessarily reduces polarization. They tackled optimization problems aimed at minimizing or maximizing polarization by adjusting edge weights within predefined constraints. Chitra and Musco~\cite{chitra2020analyzing} investigated the filter bubble effect, revealing that attempts by network administrators to mitigate disagreement and enhance user engagement inadvertently escalate polarization. Furthermore, Bhalla et al.~\cite{bhalla2023local} demonstrated in a model incorporating confirmation bias and friend-of-friend link recommendations that opinion dynamics naturally lead to polarization. For more related work on efficiently approximating polarization and disagreement, refer to~\cite{neumann2024sublinear,xu2022effects}

% For efficiently approximating the $\mathcal{PD}$ index, refer to~\cite{neumann2024sublinear,xu2022effects}. 
\textbf{Stubbornness.} The impact of stubbornness of users on the opinion dynamic process has been studied in various contexts, cf.~\cite{das2014modeling,hunter2022optimizing,n2020rumor}. For example, the effectiveness of stubborn users in resisting engineered dominant outputs in the majority model has been demonstrated by Out et al.~\cite{out2021majority}. The problem of adjusting stubbornness (in terms of susceptibility) of nodes to either maximize or minimize the total sum of opinions at equilibrium in DeGroot opinion dynamics has been extensively investigated in studies by Abebe et al.~\cite{abebe2018opinion,abebe2021opinion}. The closest work to ours is the one by Xu et al.~\cite{xu2022effects}, who analyzed the incorporation of stubbornness into the FJ opinion dynamics. They demonstrated that stubbornness significantly influences convergence time and the expressed opinion, characterized by stubbornness-dependent spanning diverging forests.

\section{Matrix Formulation and Bounds} \label{MAPD_section}

We express the polarization, disagreement, and the $\mathcal{PD}$ indices in matrix forms, in the presence of stubbornness factors. This comes very handy in the rest of the paper. Additionally, we provide upper bounds on $\mathcal{PD}$ in both the homogeneous and inhomogeneous cases in Theorems~\ref{bound_homogeneous} and \ref{bound_inhomogeneous}, respectively.

For the sake of simplicity, we define $\Vec{\textbf{1}}_\textbf{K}=\textbf{K} (\textbf{L}+\textbf{K})^{-1}  \Vec{\textbf{1}}$. % and note that if we set $\textbf{K}=\textbf{I}$, then $\Vec{\textbf{1}}_{\textbf{K}}=\Vec{\textbf{1}}$ (as $1$ is the eigenvalue of $\textbf{I}+\textbf{L}$ and the related eigenvector is $\Vec{\textbf{1}}$). 
Let's start with the following key lemma. 
\begin{lemma}\label{sbar}
The mean-centered of $z^*$ is given by $\bar{z}=(\textbf{L}+\textbf{K})^{-1} \textbf{K}\bar{s}_{\scriptscriptstyle\textbf{K}}$ where
\begin{equation}\label{s_bar_k}
\bar{s}_{\scriptscriptstyle\textbf{K}}:=s-\frac{s^T \Vec{\textbf{1}}_\textbf{K}} {{n}} \Vec{\textbf{1}}. 
\end{equation}
Moreover, if $\textbf{K}=\textbf{I}$, then 
\begin{equation}\label{s_bar_I}
\bar{s}:=s-\frac{s^T \Vec{\textbf{1}}} {{n}} \Vec{\textbf{1}}. 
\end{equation}
\end{lemma}
\begin{proof}
Using $\textbf{L} \Vec{\textbf{1}}=\Vec{\textbf{0}}$, where $\Vec{\textbf{0}}$ is the vector of size $n$ and all entries are zero, we will have $(\textbf{L}+\textbf{K})\Vec{\textbf{1}}=\textbf{K}\Vec{\textbf{1}}$. Since $\textbf{L}+\textbf{K}$ is strictly positive definite matrix, it is invertible, which implies $\Vec{\textbf{1}}=(\textbf{L}+\textbf{K})^{-1} \textbf{K} \Vec{\textbf{1}}$. %Multiplying both sides by $\textbf{K}$ gives us 
%\begin{equation}
%\overrightarrow{k}=\textbf{K}(\textbf{L}+\textbf{K})^{-1} \overrightarrow{k}, 
%\end{equation}
%which implies that \ahad{I'm not following this next equation!}
%\begin{equation}
%s^T \overrightarrow{k}= s^T \textbf{K}(\textbf{L}+\textbf{K})^{-1} \overrightarrow{k}={z^*}^T \overrightarrow{k}.
%\end{equation}
Following the definition of $\bar{z}$ and Lemma~\ref{equlibrium}, we will have 
\begin{align*}
\bar{z} &= z^* - \frac{{z^*}^T \Vec{\textbf{1}}}{n} \Vec{\textbf{1}} = (\textbf{L} + \textbf{K})^{-1} \textbf{K} \mathbf{s} - \frac{s^T \textbf{K} (\textbf{L} + \textbf{K})^{-1} \Vec{\textbf{1}}}{n} \Vec{\textbf{1}}, \\
&= (\textbf{L} + \textbf{K})^{-1} \textbf{K} s - (\textbf{L} + \textbf{K})^{-1} \textbf{K} \Vec{\textbf{1}} \frac{s^T \textbf{K} (\textbf{L} + \textbf{K})^{-1} \Vec{\textbf{1}}}{n}, \\
&= (\textbf{L} + \textbf{K})^{-1} \textbf{K} \bigg(s - \frac{s^T \Vec{\textbf{1}}_{\textbf{K}}}{n} \Vec{\textbf{1}}\bigg) = (\textbf{L} + \textbf{K})^{-1} \textbf{K} \bar{s}_{\scriptscriptstyle\textbf{K}}.
\end{align*}
The second part of the lemma will be proven by using the fact that $(\textbf{L}+\textbf{I})^{-1}\Vec{\textbf{1}}=\Vec{\textbf{1}}$.
\end{proof}

Now, we simplify the quantities of polarization, disagreement and the $\mathcal{PD}$. Following Lemma~\ref{sbar}, $\bar{z}=(\textbf{L}+\textbf{K})^{-1} \textbf{K}\bar{s}_{\scriptscriptstyle\textbf{K}}$, so the polarization index can be given as 

% \ahad{What is $\mathcal{P}_{G,\bar{z}}$? Shouldn't this be $\mathcal{P}_{G,z^*}$. All definitions in the Preliminaries section still have $z$ instead of $z^*$. As we discussed you should state there when $z=z^*$, which is what we have in our setup, then we use simply $\mathcal{P}$, $\mathcal{D}$, and $\mathcal{PD}$.}

% \ahad{Generally, we defined all these parameters for $z$ and now use them for $\bar{z}$ with any trace of $z^*$ and the connection.}

\begin{equation}\label{polarization}
\mathcal{P}:=\mathcal{P}_{G,z^*}=\bar{s}_{\scriptscriptstyle\textbf{K}}^T \textbf{K} (\textbf{L}+\textbf{K})^{-2} \textbf{K} \bar{s}_{\scriptscriptstyle\textbf{K}}.
\end{equation}
Also, the disagreement index can be computed as 
\begin{equation}\label{disagreement}
\mathcal{D}:=\mathcal{D}_{G,z^*}=\bar{s}_{\scriptscriptstyle\textbf{K}}^T \textbf{K} (\textbf{L}+\textbf{K})^{-1} \textbf{L} (\textbf{L}+\textbf{K})^{-1} \textbf{K} \bar{s}_{\scriptscriptstyle\textbf{K}}.
\end{equation}
Hence the $\mathcal{PD}$ of graph $G$ with Laplacian matrix $\textbf{L}$ and stubbornness matrix $\textbf{K}$ , innate opinions $s$  can be simplified to
\begin{equation}\label{PD_first}
     \mathcal{PD}:=\mathcal{P}_{G,z^*}+ \mathcal{D}_{G,z^*}=\bar{s}_{\scriptscriptstyle\textbf{K}}^T \textbf{K} (\textbf{K}+\textbf{L})^{-1} (\textbf{L}+\textbf{I}) (\textbf{K}+\textbf{L})^{-1} \textbf{K} \bar{s}_{\scriptscriptstyle\textbf{K}}.
\end{equation}

For the sake of simplicity, we will use the notation $\textbf{P}:=(\textbf{L}+\textbf{I})^{1/2} (\textbf{K}+\textbf{L})^{-1} \textbf{K}$ throughout the paper. Note that 
%
%\begin{lemma}\label{symmetric_pd}
the matrix $\textbf{K} (\textbf{K}+\textbf{L})^{-1} (\textbf{L}+\textbf{I}) (\textbf{K}+\textbf{L})^{-1} \textbf{K}$ 
%\end{lemma}
%\begin{proof}
can be re-written as 
\begin{equation*}
\textbf{K} (\textbf{K}+\textbf{L})^{-1} (\textbf{L}+\textbf{I})^{1/2} (\textbf{L}+\textbf{I})^{1/2} (\textbf{K}+\textbf{L})^{-1} \textbf{K}  =\textbf{P}^T \textbf{P}. 
\end{equation*}
%where $\textbf{P}:=(\textbf{L}+\textbf{I})^{1/2} (\textbf{K}+\textbf{L})^{-1} \textbf{K}$. 
Hence, it is a symmetric and positive semi-definite matrix as for any $ x \in \mathbb{R}^n$, 
%\begin{equation*}
    $x^T \textbf{P}^T \textbf{P} x^T=(\textbf{P} x)^T(\textbf{P} x)=\big\|\textbf{P} x \big\|_2^2 \geq 0.$
%\end{equation*}
%\end{proof}

In the following theorem, we will provide an upper bound (in the homogeneous case) for the $\mathcal{PD}$ based on the maximum length of the innate opinions and the stubbornness factor.

\begin{theorem}\label{bound_homogeneous}
For any graph \( G \), stubbornness matrix \( \textbf{K} = \alpha \textbf{I} \), $\alpha>2$, and any innate opinion vector $ s \in \mathbb{R}^n$ where $\| s \| \leq R$, the $\mathcal{PD}$ is upper bounded by \( \frac{R^2 \alpha^2}{4(\alpha - 1)} \). If \( 0 < \alpha < 2 \), then the upper bound is given by \( R^2 \).
\end{theorem}

\begin{proof}
If we set $\textbf{K}=\alpha \textbf{I}$, then
\begin{equation*}
\bar{s}_{\scriptscriptstyle\textbf{K}}=s-\frac{s^T \textbf{K} (\textbf{L}+\textbf{K})^{-1} \Vec{\textbf{1}}}{n} \Vec{\textbf{1}}=s-\frac{s^T \Vec{\textbf{1}}}{n} \Vec{\textbf{1}}=\bar{s},
\end{equation*}
as $\textbf{K}(\textbf{L}+\textbf{K})^{-1} \Vec{\textbf{1}}=(\textbf{I}+\alpha^{-1}\textbf{L})^{-1} \Vec{\textbf{1}}=\Vec{\textbf{1}}$. So, the $\mathcal{PD}$ will be simplified as  
$$\mathcal{PD}=\bar{s}^T (\textbf{I}+\alpha^{-1} \textbf{L})^{-1} (\textbf{I}+ \textbf{L})(\textbf{I}+\alpha^{-1} \textbf{L})^{-1} \bar{s},$$
as $(\textbf{L}+\textbf{K})^{-1}=\alpha^{-1}(\textbf{I}+\alpha ^{-1}\textbf{L})^{-1}$.
% If we set $\textbf{K}=\alpha \textbf{I}$, then $\mathcal{PD}=\bar{s}^T (\textbf{I}+\alpha^{-1} \textbf{L})^{-1} (\textbf{I}+ \textbf{L})(\textbf{I}+\alpha^{-1} \textbf{L})^{-1} \bar{s}$. Also as $\Vec{\textbf{1}}=(\textbf{I}+\alpha^{-1} \textbf{L})^{-1} \Vec{\textbf{1}}$, we have 
% \begin{equation*}
% \bar{s}_{\scriptscriptstyle\textbf{K}}=s-\frac{s^T \textbf{K} (\textbf{L}+\textbf{K})^{-1} \Vec{\textbf{1}}}{n} \Vec{\textbf{1}}=s-\frac{s^T \Vec{\textbf{1}}}{n} \Vec{\textbf{1}}=\bar{s}.
% \end{equation*}
% So $\bar{s}$ will have the following expansion 
% %\begin{equation}\label{s_bar_expansion}
%     $\bar{s}=\sum_{j=2}^n \gamma_j q_j.$
% %\end{equation}
Since $\textbf{L}$ is a symmetric positive semi-definite, 
\begin{equation*}    
F(\textbf{L},\alpha):=(\textbf{I}+\alpha^{-1}\textbf{L})^{-1} (\textbf{I}+\textbf{L}) (\textbf{I}+\alpha^{-1}\textbf{L})^{-1}=\sum_{j=1}^n \frac{1+\lambda_j}{(1+\alpha^{-1}\lambda_j)^2} q_j q_j^T.  
\end{equation*}
As any internal opinion $s \in \mathbb{R}^n$, with $ \| s \| \leq R$, can be written as $\sum_{j=1}^n \gamma_j q_j$, for some coefficients $\gamma_j$, we will have 
\begin{equation*}
\bar{s}^T F(\textbf{L},\alpha) \bar{s}=\sum_{j=2}^n \frac{1+\lambda_j}{(1+\alpha^{-1}\lambda_j)^2} \gamma_j^2 \leq \max_{j=2,\cdots,n} \bigg( \frac{1+\lambda_j}{(1+\alpha^{-1}\lambda_j)^2} \bigg) R^2. 
\end{equation*}
Note that the sum starts from index $2$ and the reason is that $\Vec{\textbf{1}}$ is the eigenvector of Laplacian matrix and also $\langle s, \Vec{\textbf{1}} \rangle =0$. 
% By the variational characterization of eigenvalues of symmetric matrices, for any innate opinion vector $s$ where $\|s\| \leq R$, 
% \begin{equation*}
% \max_{\|s \| \leq R} \bar{s}^T (\textbf{I}+\alpha^{-1} \textbf{L})^{-1} (\textbf{I}+ \textbf{L})(\textbf{I}+\alpha^{-1} \textbf{L})^{-1} \bar{s}, 
% \end{equation*}
% \begin{equation*}
% =\max_{j=2,\cdots,n} \bigg( \frac{1+\lambda_j(\textbf{L})}{\big(1+\alpha^{-1} \lambda_j(\textbf{L})\big)^2} \bigg)  \| s \|_2^2. 
% \end{equation*}
Now, consider $g(x)$ as $g(x):=\frac{1+x}{(1+\alpha^{-1}x)^2}$. It is easy to compute that 
%\begin{equation*}
$g'(x):=\frac{\alpha-2 -x}{\alpha (1+\alpha^{-1}x)^3}$, 
%\end{equation*}
which increases if and only if $x <\alpha-2$ and decreases if and only if $x> \alpha-2$, attaining a pick of $\frac{\alpha^2}{4(\alpha-1)}$ at $x=\alpha-2$ which immediately gives the upper bound. If \( 0 < \alpha <2 \), then as \( x = \alpha - 2 < 0 \), \( g(x) \) achieves its maximum \( R^2 \) at \( x = 0 \).
\end{proof}

In the following theorem, we will provide an upper bound for the $\mathcal{PD}$ in the inhomogeneous case.

\begin{theorem}\label{bound_inhomogeneous}
For any graph $G$ with given stubbornness matrix $\textbf{K}$ and any innate opinion vector $ s \in \mathbb{R}^n$ where $\| s \| \leq R$, the $\mathcal{PD}$ is upper bounded by
\begin{equation*}
\big(R^2+ \mu^2 n \big) \lambda_{\max}(\textbf{P}^T\textbf{P}), 
\end{equation*}
where $\mu:=\max_{s \in \mathbb{R}^n} \frac{s^T ( \Vec{\textbf{1}}-\Vec{\textbf{1}}_{\textbf{K}})}{n}$.
%for some constant $\mu$.%where $\mu:=\max_{s \in \mathbb{R}} \frac{s^T ( \Vec{\textbf{1}}-\Vec{\textbf{1}}_{\textbf{K}})}{n}$ is a real number. 
\end{theorem}
\begin{proof}
As $\textbf{P}^T \textbf{P}$ is symmetric positive semi-definite, it can be written as follows, cf.~\cite{golub2013matrix}:
\begin{equation*}
    \textbf{P}^T \textbf{P}=\sum_{j=1}^n \lambda_j p_j p_j^T, 
\end{equation*}
where $p_j$ is the eigenvector (of matrix $\textbf{P}^T \textbf{P}$) related to the eigenvalue $\lambda_j$. As $\{p_j\}_{j=1}^n$ forms an orthonormal basis for $\mathbb{R}^n$, any $\bar{s}_{\scriptscriptstyle\textbf{K}} \in \mathbb{R}^n$ can be written as $\bar{s}_{\scriptscriptstyle\textbf{K}}=\sum_{j=1}^n \alpha_j p_j$ where $\alpha_j \in \mathbb{R}$. Hence, as $\| \bar{s}_{\scriptscriptstyle\textbf{K}} \|^2=\sum_{j=1}^n \alpha_j^2$ we will have 
\begin{equation}\label{upper_bound_MAPD}
%\max_{s \in \mathbb{R}^n, \|s\| \leq R} 
\bar{s}_{\scriptscriptstyle\textbf{K}}^T \textbf{P}^T \textbf{P} \bar{s}_{\scriptscriptstyle\textbf{K}}=\sum_{j=1}^n \alpha_j^2 \lambda_j \leq \lambda_{\max}(\textbf{P}^T \textbf{P}) R^*, \quad \forall\ \|s\| \leq R, 
\end{equation}
where $R^*= \max_{s \in \mathbb{R}^n}\| \bar{s}_{\scriptscriptstyle\textbf{K}} \|_2^2$. 
% Now, without loss of generality, assume that \(\lambda_{\max}(\textbf{P}^T\textbf{P})\) is achieved for the \(l\)-th element and choose \(s \in \mathbb{R}^n\) such that \(\| \bar{s}_{\scriptscriptstyle\textbf{K}} \|\) is maximized. Observe that the upper bound in equation (\ref{upper_bound_MAPD}) is achievable for the vector \(s^* = (0, \ldots, \sqrt{R^*}, \ldots, 0)\), where $R^*$ is located at the \(l\)-th element. So the $\mathcal{MPD}$ will be given as
% \ahad{This doesn't see, to be equality, but rather an upper bound. If that's the case, your results from below won't work anymore.}
% \begin{equation*}
% \lambda_{\max}(\textbf{P}^T \textbf{P}) R^*.  
% \end{equation*}
Following Lemma~\ref{sbar},   
%\begin{equation*}
$\bar{s}_{\scriptscriptstyle\textbf{K}}=\bar{s}+\frac{s^T \big( \Vec{\textbf{1}}-\Vec{\textbf{1}}_{\textbf{K}} \big)}{n} \Vec{\textbf{1}}$
%\end{equation*}
and if we set $\mu:=\max_{s \in \mathbb{R}^n} \frac{s^T ( \Vec{\textbf{1}}-\Vec{\textbf{1}}_{\textbf{K}})}{n}$, we will have 
\begin{equation*}
\| \bar{s}_{\scriptscriptstyle\textbf{K}} \|_2^2=\langle \bar{s}, \bar{s} \rangle +2 \langle \bar{s}, \mu \Vec{\textbf{1}} \rangle +\langle \mu \Vec{\textbf{1}},\mu \Vec{\textbf{1}}  \rangle,
\end{equation*}
which gives us $R^*=R^2+\mu^2 n$ as $\langle \bar{s}, \Vec{\textbf{1}} \rangle=0$ and the proof is complete. 
\end{proof}

% In the following theorem, we will demonstrate that an increased stubbornness coefficient leads to a higher $\mathcal{MPD}$ in comparison with the classic FJ model.

% \begin{theorem}\label{MAPD}
% The $\mathcal{MPD}$ when $k_i \geq 1$, $i=1,\cdots,n$ and at least for one node $k_j>1$, is larger than $\mathcal{MPD}$ when all $k_i=1$ for all $i=1,\cdots,n$.
% \end{theorem}
% \begin{proof}
% Following Corollary~\ref{co_max_PD_K_alpha_I}, the $\mathcal{MPD}$ when all $k_i=1$ is $\frac{R^2}{1+\lambda_2(\textbf{L})}$, which is strictly less than $R^2$ since $\lambda_2(\textbf{L}) > 0$. However, since $\mu^2 n \geq 0$, the maximum achievable PD value when at least one node's stubbornness is greater than or equal to one will exceed $R^2 \lambda_{\max} (\textbf{P}^T \textbf{P})$, which itself is greater than or equal to $R^2$ following Lemma~\ref{largest_eigenvalue_P^TP}.
% \end{proof}

\section{$\mathcal{PD}$ in homogeneous stubbornness} \label{homogeneous_section}
% In practice, individuals exhibit varying degrees of stubbornness.
In this section, we consider the homogeneous setup, where all nodes have the same stubbornness $\alpha>0$, that is, $\textbf{K} = \alpha \textbf{I}$. This facilitates a rigorous analysis of the impact of stubbornness. Our main result in this section is Theorem~\ref{P_PD_K_alpha_I_theorem} which demonstrates that any increase in the stubbornness coefficient of the network leads to an increase in $\mathcal{PD}$.  
% In the following theorem, we outline the conditions under which polarization and the $\mathcal{PD}$ index increase or decrease.
For the sake of simplicity, let $\mathcal{PD}(\alpha)$ and $\mathcal{P}(\alpha)$ denote the $\mathcal{PD}$ and the polarization with stubbornness matrix $\textbf{K}=\alpha \textbf{I}$. 

\begin{theorem}\label{P_PD_K_alpha_I_theorem}
Consider a graph \( G = (V, E) \) and innate opinion vector $s \in  \mathbb{R}^n$. Then, for any  $\alpha \leq \beta$, 
\begin{equation*}
     \mathcal{PD}(\alpha) \leq  \mathcal{PD}(\beta), \quad \mathcal{P}(\alpha) \leq \mathcal{P}(\beta). 
\end{equation*}
%and 
%\begin{equation*}
%\max_{s \in \mathbb{R}, \| s \| \leq R} \big(\PD(\beta)-\PD(\alpha)  \big)=\bigg(\frac{1}{\left(\beta^{-1} C- 1\right)^2} - \frac{1}{\left(\alpha^{-1} C- 1\right)^2}\bigg) R^2
%\end{equation*}
%where 
%\begin{equation*}
%C:=\frac{(\frac{\beta}{\alpha})^{1/3} - 1}{\left(\beta^{-1}(\frac{\beta}{\alpha})^{1/3}- \alpha^{-1}\right)}
%\end{equation*}
\end{theorem}
\begin{proof}
First, note that $(\textbf{L}+\textbf{K})^{-1}=(\textbf{L}+\alpha \textbf{I})^{-1}=\alpha^{-1} (\textbf{I}+\alpha^{-1}\textbf{L})^{-1}$. 
As $(\textbf{I}+\alpha^{-1} \textbf{L}) \Vec{\textbf{1}}=\Vec{\textbf{1}}$, we will have $\Vec{\textbf{1}}=(\textbf{I}+\alpha^{-1} \textbf{L})^{-1} \Vec{\textbf{1}}$, and hence 
\begin{equation*}
\bar{s}_{\scriptscriptstyle\textbf{K}}=s-\frac{s^T \textbf{K} (\textbf{L}+\textbf{K})^{-1} \Vec{\textbf{1}}}{n} \Vec{\textbf{1}}=s-\frac{s^T \Vec{\textbf{1}}}{n} \Vec{\textbf{1}}=\bar{s}.
\end{equation*}
Indeed, as the stubbornness is homogeneous,  polarization and the $\mathcal{PD}$ given in Equations~\eqref{polarization} and \eqref{PD_first} will take the following forms %\ahad{Are you using $\bar{s}_{\scriptscriptstyle\textbf{K}} = \bar{s}$? This mentioned in the proof of Theorem~\ref{bound_homogeneous}, but I doubt if the reader would know or remember. You need to mention it. (DONE)}
\begin{equation}\label{P_L_s_K}
    \mathcal{P}(\alpha)=\alpha^2 \bar{s}^T (\textbf{L}+\textbf{K})^{-2} \bar{s}= \bar{s}^T (\textbf{I}+\alpha^{-1} \textbf{L})^{-2} \bar{s}, 
\end{equation}
\begin{equation}\label{PD_L_s_K}
    \mathcal{PD}(\alpha)=\alpha^2 \bar{s}^T (\textbf{L}+\textbf{K})^{-1} (\textbf{I}+\textbf{L}) (\textbf{L}+\textbf{K})^{-1} \bar{s}, 
\end{equation}
\begin{equation*}
\quad \quad \quad \quad \quad = \bar{s}^T (\textbf{I}+\alpha^{-1}\textbf{L})^{-1} (\textbf{I}+\textbf{L}) (\textbf{I}+\alpha^{-1} \textbf{L})^{-1} \bar{s}. 
\end{equation*}
 %\ahad{I'm not sure what is the purpose of this last sentence. It isn't even a full sentence.}. 
% Note that as $(\textbf{I}+\alpha^{-1} \textbf{L}) \Vec{\textbf{1}}=\Vec{\textbf{1}}$, we will have $\Vec{\textbf{1}}=(\textbf{I}+\alpha^{-1} \textbf{L})^{-1} \Vec{\textbf{1}}$, and 
% \begin{equation*}
% \bar{s}_{\scriptscriptstyle\textbf{K}}=s-\frac{s^T \textbf{K} (\textbf{L}+\textbf{K})^{-1} \Vec{\textbf{1}}}{n} \Vec{\textbf{1}}=s-\frac{s^T \Vec{\textbf{1}}}{n} \Vec{\textbf{1}}=\bar{s}.
% \end{equation*}
%\mohammad{I added}
%\mohammad{I added this in the beginning of the proof}. 
Also, note that for any $s \in \mathbb{R}^n$, $s$ can be written as $s=\sum_{j=1}^n \gamma_j q_j^T$, where $q_j$ and $\gamma_j$, for $j=1,\cdots,n$, are the eigenvectors of matrix $\textbf{L}$ and real coefficients, respectively. Hence, $\bar{s}$ will have the following expansion 
\begin{equation*}
    \bar{s}=s-\frac{s^T \Vec{\textbf{1}}}{n} \Vec{\textbf{1}}=\gamma_1 \Vec{\textbf{1}}+\sum_{j=2}^n \gamma_j q_j-\frac{\bigg(\gamma_1 \Vec{\textbf{1}}^T+\sum_{j=2}^n \gamma_j q_j^T \bigg) \Vec{\textbf{1}}}{n} \Vec{\textbf{1}}, 
\end{equation*}
%where $q_j$, for $j=2,\cdots,n$, are the eigenvectors of matrix $\textbf{L}$ and $\gamma_j$, for $j=2,\cdots\,n$ are real coefficients where \ahad{You use ``where'' two times in a row in this sentence.} $\sum_{j=2}^n \gamma_j^2=\| \bar{s} \|_2^2$. 
which gives  $\bar{s}=\sum_{j=2}^n \gamma_j q_j$ as $\Vec{\textbf{1}}^T \Vec{\textbf{1}}=n$ and $q_i^T q_j=1$ if  $i=j$ and $0$ otherwise. 

Using Equations~\eqref{f(L)} and \eqref{multiplication_series}, Equations \eqref{P_L_s_K} and \eqref{PD_L_s_K} will have the following representations
\begin{equation*}
\mathcal{P}(\alpha)=\sum_{j=2}^n \frac{1}{(1+\alpha^{-1} \lambda_j)^2}\gamma_j^2, \quad 
 \mathcal{PD}(\alpha)=\sum_{j=2}^n \frac{1+\lambda_j}{(1+\alpha^{-1} \lambda_j)^2} \gamma_j^2.
\end{equation*}
The proof then will follow as for any $\beta \geq \alpha$, $\frac{1}{(1+\beta^{-1}x)^2} \geq \frac{1}{(1+\alpha^{-1}x)^2}$ for any $x \geq 0$.
\end{proof}

While our main focus in this paper is on the $\mathcal{PD}$, we state the following interesting result about the connection between stubbornness and polarization, which follows from a novel proof idea. The upper bound depends only on the stubbornness factors (pre- and post-increase), but not the graph structure.

\begin{theorem}\label{p_changes_upper_bound}
Consider a graph \( G = (V, E) \) and a given innate opinion vector $s \in \mathbb{R}^n$ such that $\|s\| \leq R$, then
\begin{equation*}
\mathcal{P}(\beta)-\mathcal{P}(\alpha)  \leq \bigg(\frac{1}{\left(1+\beta^{-1} C\right)^2} - \frac{1}{\left(1+\alpha^{-1} C\right)^2}\bigg) R^2,
\end{equation*}
where 
\begin{equation}\label{C}
C:%\frac{(\frac{\beta}{\alpha})^{1/3} - 1}{\left(\beta^{-1}(\frac{\beta}{\alpha})^{1/3}- \alpha^{-1}\right)}
=\frac{\alpha^{1/3}-\beta^{1/3}}{\beta^{-2/3}-\alpha^{-2/3}}. 
\end{equation}
\end{theorem}
\begin{proof}
Let's define $E(x)$ as follow 
\begin{equation*}
E(x):=f(x,\beta)-f(x,\alpha), \quad f(x,c)=\frac{1}{(1+c^{-1}x)^2}.
\end{equation*}
%\ahad{This sentence is confusing.\mohammad{red}{I changed it.}} 
The unique extremum of \( E(x) \), denoted as $C$, can be found by setting its first derivative, \( \frac{dE(x)}{dx} \), to zero, given by  
\[
C := \frac{\alpha^{1/3} - \beta^{1/3}}{\beta^{-2/3} - \alpha^{-2/3}}.
\]
It can be seen that $C$ is the maximum of $E(x)$ which implies that
\begin{equation}\label{max_E}
\max_{x >0} \big(f(x,\beta)-f(x,\alpha)\big)=E(C)=\frac{1}{\left(1+\beta^{-1} C\right)^2} - \frac{1}{\left(1+\alpha^{-1} C\right)^2}.
\end{equation}
% where 
% \begin{equation*}
% C:=\frac{(\frac{\beta}{\alpha})^{1/3} - 1}{\left(\beta^{-1}(\frac{\beta}{\alpha})^{1/3}- \alpha^{-1}\right)}.
% \end{equation*}
Finally, 
\begin{equation*}
\mathcal{P}(\beta)-\mathcal{P}(\alpha)= \bar{s}^T \bigg((\textbf{I}+\beta^{-1} \textbf{L})^{-2}-(\textbf{I}+\alpha^{-1} \textbf{L})^{-2}  \bigg) \bar{s} 
\end{equation*}
\begin{equation*}
\quad \quad \quad \quad \quad \quad \quad =\sum_{j=2}^n \bigg(\frac{1}{(1+\beta^{-1} \lambda_j)^2}-\frac{1}{(1+\alpha^{-1} \lambda_j)^2}\bigg)\gamma_j^2. 
\end{equation*}
Thus, invoking Equation~\eqref{max_E} and taking into account that $\| \bar{s} \|_2 \leq \| s \|_2 \leq R$, the proof is thereby concluded.
\end{proof}

%\ahad{We should add a $\qed$ at the end of all proofs in the appendix. \mohammad{red}{done}}

To better understand the upper bound in Theorem~\ref{p_changes_upper_bound}, we set $\alpha = 1$ and vary $\beta$ from $1$ to $10^5$ in Figure~\ref{fig:p_homogeneous_p}. We observe that the changes in the bound for small values of $\alpha$, i.e., $\beta <100$ is significant, but as $\beta$ grows the impact vanishes. This is because once polarization approaches the maximum achievable value, increasing stubbornness has negligible impact.
% We also observe that the constant stubbornness values $\alpha$ and $\beta$ in 
% Theorem~\ref{p_changes_upper_bound}, i.e., $\frac{1}{\left(1+\beta^{-1} C\right)^2} - \frac{1}{\left(1+\alpha^{-1} C\right)^2}$  is always less than $1$ and the length of the internal opinion, i.e., $R$, plays the important role in the upper bound.  

% \ahad{What exactly is this diagram. You say that the y-axis is the change in polarization, but I don't think that's the case. Is the value $\mathcal{P}(\beta)-\mathcal{P}(1)$ where $b$ is the x-axis? In any case, this needs quite some clarifications and also you need to explain what you understand from this diagram. The goal is to explain how good your proven bound is, but there is nothing about that in the text. (I change the chart based on our discussion.)}

\begin{figure}[H]
\begin{center}
    \includegraphics[width=7cm, height=5cm]{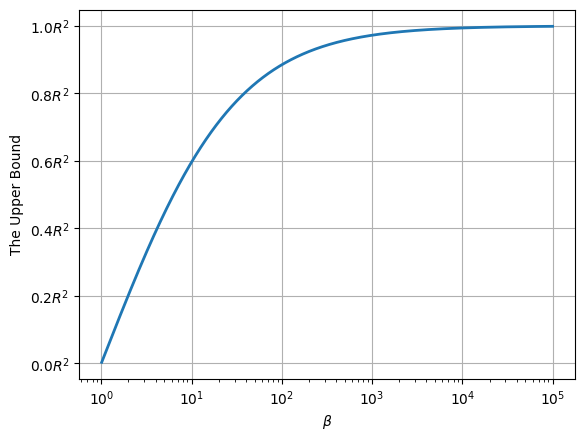}
\end{center}
\caption{Variation in the provided polarization upper bound in terms of the stubbornness factor.} \label{fig:p_homogeneous_p}
\end{figure}

\subsection{Stochastic Block Model Analysis}

One special setup which has attracted attention in the prior work, with the purpose of modeling social bubbles, is when there are two extreme communities which are not well-connected, modeled by the Stochastic Block Model (SBM), cf.~\cite{chitra2020analyzing,chitra2019understanding,bhalla2023local}. We give a full characterization of the connection between stubbornness and $\mathcal{PD}$ in this setup.

The SBM is a random graph model with parameters $n$ (consider to be an even number) and $p, q \in [0,1]$. In the studied setup, nodes are partitioned into two groups, $V_+=\{v_1,\cdots,v_{\frac{n}{2}}\}$ and $V_-=\{v_{\frac{n}{2}+1},\cdots,v_{n} \}$. Edges will be generated based on the following process. For all $v_i, v_j \in V$, if $v_i, v_j \in V_+$ or  $v_i, v_j \in V_-$, set $w_{v_iv_j}=1$ independently with probability $p$, and $w_{v_iv_j}=0$ otherwise. If $v_i \in V_+$, $v_j \in V_-$, set $w_{v_iv_j}=1$ independently with probability $q$ and $w_{v_iv_j}=0$ otherwise. We are focusing on the common scenario in SBM where $q$ is less than $p$. This means that the likelihood of two nodes being connected is higher when the nodes belong to the same community. (This is the special case considered in the prior work~\cite{chitra2020analyzing,chitra2019understanding,bhalla2023local}; otherwise, SBM is more general and allows more than two communities with different values of $p$.)

The expected adjacency matrix $\bar{\textbf{A}}$ will be given as 
\begin{equation}\label{adjacency_A}
\bar{\textbf{A}}_{ij}=
\begin{cases}
\text{p,} &\quad\text{if $i\in [1,n/2],  j \in [1,n/2]$ } \\
\text{q,} &\quad\text{if $i \in [1,n/2],  j \in [n/2+1,n]$ } \\
\text{q,} &\quad\text{if $i \in [n/2+1,n], j \in [1,n/2]$ } \\
\text{p,} &\quad\text{if $i \in [n/2+1,n], j \in [n/2+1,n]$} \\
\end{cases}
\end{equation}

Nodes in set \( V_- \) have innate opinions \(-1\), and nodes in set \( V_+ \) have innate opinions \(+1\) (this is supposed to mimic the \textit{homophily} property, which states more similar nodes are more likely to be connected). 
\begin{equation}\label{intenal_opinion_SBM}
s_i=
\begin{cases}
\text{-1,} &\quad\text{if $i\in V_-$ } \\
\text{+1,} &\quad\text{if $i \in V_+$ } \\
\end{cases}
\end{equation}

Let us make two observations: (i) $\bar{s}=s$ and (ii) $\bar{s}$ is an eigenvector of $\bar{\textbf{A}}$ with eigenvalue $\frac{(p-q)n}{2}$. Thus, $\bar{s}$ is an eigenvector of $\bar{\textbf{L}}=\frac{n}{2}(p+q)\textbf{I}-\bar{A}$ with eigenvalue $qn$. Now, we provide Theorem~\ref{SBM_theorem} (which is a generalization of the results from~\cite{chitra2020analyzing, bhalla2023local} for the vanilla FJ model).

\begin{theorem}\label{SBM_theorem}
Let $G=(V, E)$ be the SBM graph described above with the expected adjacency matrix defined in Equation~\eqref{adjacency_A} and innate opinion from Equation~\eqref{intenal_opinion_SBM}, and consider the stubbornness matrix $\textbf{K}=\alpha \textbf{I}$. Then, we have
\begin{equation*}
%P(\bar{\textbf{L}},\bar{s},\textbf{K})=\frac{\alpha^2 n}{(nq+\alpha)^2}, \quad \quad 
   \mathcal{PD}(\alpha)=\frac{\alpha^2(1+nq)n}{(nq+\alpha)^2}. 
\end{equation*}
\end{theorem}

\begin{proof}
As $\textbf{K}=\alpha \textbf{I}$, we will have $\bar{s}_{\scriptscriptstyle\textbf{K}}=\bar{s}$, and the $\mathcal{PD}$, which is given in Equation~\eqref{PD_L_s_K}, can be written as  
\begin{equation*}
(\textbf{I}+\alpha^{-1}\textbf{L})^{-1} (\textbf{I}+\textbf{L}) (\textbf{I}+\alpha^{-1} \textbf{L})^{-1}=\sum_{j=1}^n \frac{1+\lambda_j}{(1+\alpha^{-1} \lambda_j)^2} q_j q_j^T, 
\end{equation*}
where $(\lambda_j,q_j)$, for $j=1,\cdots,n$ is the eigenpair of the Laplacian matrix $\textbf{L}$. Now, the $\mathcal{PD}$ will have the following form 
\begin{equation*}
\mathcal{PD}=\sum_{j=2}^n \frac{1+\lambda_j}{(1+\alpha^{-1} \lambda_j)^2} \bar{s}^T q_j q_j^T \bar{s},
\end{equation*}
\begin{equation*}
\quad \quad \quad \quad \quad \quad \quad \quad =\sum_{j=2}^n \frac{1+\lambda_j}{(1+\alpha^{-1} \lambda_j)^2} (\| \bar{s} \|\frac{\bar{s}}{\| \bar{s} \|})^T q_j q_j^T (\| \bar{s} \|\frac{\bar{s}}{\| \bar{s} \|}). 
\end{equation*}
Given that $(q_n, \bar{s})$ is an eigenpair of $\mathbf{L}$ and utilizing the orthonormality property of the eigenvectors, the expression $q_j^T \frac{\bar{s}}{\| \bar{s} \|}$ equals one if and only if $q_j=\frac{\bar{s}}{\| \bar{s} \|}$ and zero otherwise. Consequently, the $\mathcal{PD}$ will be given as 
\begin{equation*}
\mathcal{PD}= \frac{1+qn}{(1+\alpha^{-1} qn)^2} \| \bar{s} \|^2, 
\end{equation*}
the proof is then concluded as $\| \bar{s} \|^2=n$. 
\end{proof}

According to Theorem~\ref{SBM_theorem},
%When the value of $n$ is sufficiently large, polarization will scale as approximately $\frac{\alpha^2}{nq^2}$. This means that the polarization of an SBM graph will decrease quadratically with the probability of out-group connection, $q$, and increase quadratically with the stubbornness factor, $\alpha$. As for the $\mathcal{PD}$ index, its 
the $\mathcal{PD}$'s asymptotic behavior is given by $\frac{\alpha^2}{q}$, which decreases linearly with the outgoing connection probability and increases quadratically with the stubbornness factor. Surprisingly, the value of $\mathcal{PD}$ doesn't depend on the value of $p$ (as far as our assumption $p>q$ is satisfied). Figure~\ref{fig:sbm_p_pd} illustrates the quadratic relationship between stubbornness and the $\mathcal{PD}$ index.

\begin{figure}[H]
\centering
\includegraphics[width=7cm, height=5cm]{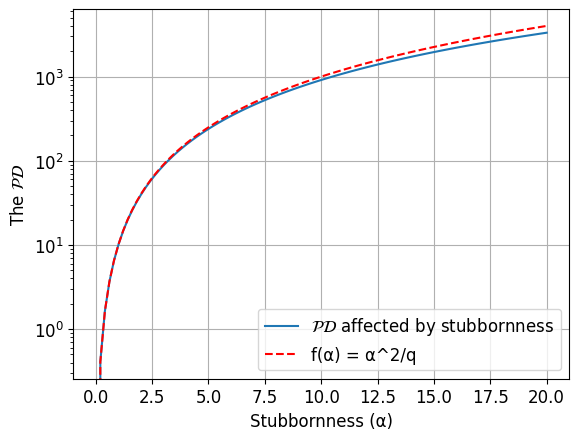}
\caption{Variation in $\mathcal{PD}$ (in logarithmic scale) as a function of network's stubbornness, within the SBM (block size = $1000$, outgoing connection probability $q = 0.1$). %\ahad{Shouldn't the value be exactly 0 when $\alpha=0$? }\ahad{Then I suggest to manually use 0 there or just remove $10^{-3}$ because the current version is incorrect. \mohammad{red}{This is fixed}}
}
\label{fig:sbm_p_pd}
\end{figure}

\section{$\mathcal{PD}$ in inhomogeneous stubbornness} \label{inhomogeneous_section}
So far, our results confirm the natural guess that an increase in stubbornness intensifies the $\mathcal{PD}$. However, surprisingly, this is not always the case as observed in the example below.

\textbf{Example:} \textit{Increase in stubbornness can cause reduction in the $\mathcal{PD}$!} Consider a graph with three nodes ($A$ connected to $B$ and $B$ connected to $C$) with innate opinions $s_A=+1, s_B=-1$ and $s_C=0$. First, consider the scenario where they all have the same level of stubbornness $k_A=k_B=k_C=1$. In this setup, the $\mathcal{PD}$ is calculated to be $0.6250$. Now, if we increase the stubbornness of node $C$ from $k=1$ to $k=2$ (while keeping all other assumptions the same), the $\mathcal{PD}$ index in this scenario will decrease to $0.6075$. The reduction in $\mathcal{PD}$ does not occur solely when the innate opinion of the updated node is zero. Figure~\ref{fig:color_reduction} illustrates that within the entire range of $(-0.31,0.57)$, the $\mathcal{PD}$ decreases as the stubbornness of node $C$ increases from $k=1$ to $k=2$.

\begin{figure}[H]
\centering
%\includegraphics[width=4.2cm, height=3.5cm]{Polarization_SBM_n_100_q_0.1}
% \hspace{0.5cm}
\includegraphics[width=7cm, height=5cm]{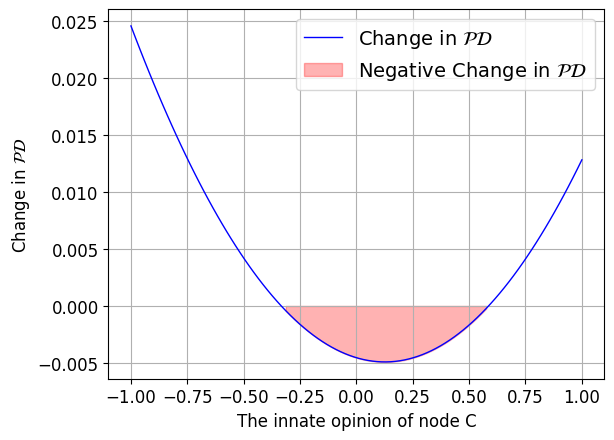}
\caption{The range of innate opinion of node $C$, for which increasing stubbornness from $k=1$ to $k=2$ causes reduction in the $\mathcal{PD}$.}
\label{fig:color_reduction}
\end{figure}

We will demonstrate that an $\epsilon$ increase in the stubbornness coefficient of a \textit{neutral} individual (a node with a zero innate opinion) in the FJ model leads to a decrease in the $\mathcal{PD}$ index. Let us establish some foundation before presenting our main results in Theorem~\ref{thm:neutral}.

\begin{lemma}\label{P1=1}
The vector $\Vec{\textbf{1}}$ is an eigenvector of the matrix $\textbf{P}$ with an eigenvalue of one, i.e., $\textbf{P} \Vec{\textbf{1}} = \Vec{\textbf{1}}$.
\end{lemma}
\begin{proof}
Note that 
%\begin{equation*}
    $\big((\textbf{L}+\textbf{K})^{-1} \textbf{K}\big)^{-1}=\textbf{K}^{-1} \textbf{L} + \textbf{I}$.
%\end{equation*}
As $\textbf{L} \Vec{\textbf{1}}=0$, $(\textbf{I}+\textbf{K}^{-1} \textbf{L}) \Vec{\textbf{1}} = \Vec{\textbf{1}}$. Hence, $\Vec{\textbf{1}} = (\textbf{I}+\textbf{K}^{-1} \textbf{L})^{-1} \Vec{\textbf{1}} $. As the eigenvector of $\textbf{L}$ and $(\textbf{I}+\textbf{L})^{1/2}$ are the same, $\Vec{\textbf{1}}$ is the eigenvector of matrix $(\textbf{L}+\textbf{I})^{1/2}$ and the related eigenvalue is $1$, so $\Vec{\textbf{1}}$ is the eigenvector of matrix $\textbf{P}$ as 
\begin{equation*}
    \textbf{P}  \Vec{\textbf{1}} =\bigg((\textbf{L}+\textbf{I})^{1/2} (\textbf{L}+\textbf{K})^{-1} \textbf{K}\bigg) \Vec{\textbf{1}}= \bigg((\textbf{L}+\textbf{I})^{1/2} (\textbf{I}+\textbf{K}^{-1} \textbf{L})^{-1} \bigg) \Vec{\textbf{1}}= \Vec{\textbf{1}}. 
\end{equation*}
\end{proof}

In the following lemma, we derive a formula for the $\mathcal{PD}$ based on $s$ rather than $\bar{s}_{\scriptscriptstyle\textbf{K}}$.

\begin{lemma}\label{PD_hat}
Let \(\widehat{\mathcal{PD}}\) be the polarization-disagreement when the stubbornness of a specific node \( l \) increases by \(\epsilon\), then for any mean-centered internal opinion vector $s \in \mathbb{R}^n$, 
\begin{equation*}
\widehat{\mathcal{PD}}=\bar{s}^T \textbf{P}^T \textbf{P} \bar{s}-\frac{1}{n}  \langle s, \Vec{\textbf{1}}_{\textbf{K}} \rangle^2.
\end{equation*}
\end{lemma}
\begin{proof}
% In the coming part, we will demonstrate the relationship between the polarization-disagreement when the stubbornness of a specific node \( l \) is increased by \(\epsilon\), denoted as \(\widehat{\PD}\),  and the polarization-disagreement in the original Friedkin-Johnsen (FJ) model, denoted by \(\PD \). 
First, note that following Lemma~\ref{sbar}, $\bar{s}_{\scriptscriptstyle\textbf{K}}=\bar{s}+\mu \Vec{\textbf{1}}$ where $\mu=\frac{s^T (\Vec{\textbf{1}}-\Vec{\textbf{1}}_{\textbf{K})}}{n}$, and so \(\widehat{\mathcal{PD}}\) can be written as 
\begin{equation*}
\widehat{\mathcal{PD}}=\big(\bar{s}+\mu \Vec{\textbf{1}}\big)^T \textbf{P}^T \textbf{P} \big(\bar{s}+\mu \Vec{\textbf{1}}\big)
\end{equation*}
\begin{equation*}
\quad \quad \quad \quad \quad \quad \quad \quad \quad =\bar{s}^T \textbf{P}^T \textbf{P} \bar{s}+\mu \bar{s}^T \textbf{P}^T \textbf{P} \Vec{\textbf{1}}+\mu \Vec{\textbf{1}}^T \textbf{P}^T \textbf{P} \bar{s}+\mu^2 \Vec{\textbf{1}}^T \textbf{P}^T \textbf{P} \Vec{\textbf{1}}. 
\end{equation*}
which can be simplified as 
\begin{equation*}
\widehat{\mathcal{PD}}=\bar{s}^T \textbf{P}^T \textbf{P} \bar{s}+2 \mu \langle \bar{s}, \Vec{\textbf{1}}_{\textbf{K}} \rangle+\mu^2 n,
\end{equation*}
where we used the fact that $\textbf{P} \Vec{\textbf{1}}=\Vec{\textbf{1}}$ as we proved in Lemma~\ref{P1=1} and also the fact that $(\textbf{I}+\textbf{L})^{1/2} \Vec{\textbf{1}}=\Vec{\textbf{1}}$. Now, following our assumption that the innate opinion vector has mean zero, we will have 
$\mu=\frac{1}{n} \big(\langle s, \Vec{\textbf{1}} \rangle - \langle s, \Vec{\textbf{1}}_{\textbf{K}} \rangle \big)=-\frac{1}{n} \langle s, \Vec{\textbf{1}}_{\textbf{K}} \rangle $ which completes the proof as  
\begin{equation*}
2 \mu \langle \bar{s}, \Vec{\textbf{1}}_{\textbf{K}} \rangle+\mu^2 n=-\frac{2}{n}  \langle s, \Vec{\textbf{1}}_{\textbf{K}} \rangle^2+\frac{1}{n}  \langle s, \Vec{\textbf{1}}_{\textbf{K}} \rangle^2=-\frac{1}{n}  \langle s, \Vec{\textbf{1}}_{\textbf{K}} \rangle^2.
\end{equation*}
\end{proof}

% The theorem below is the central result of this section. It explains the scenario in which an increase in stubbornness leads to a decrease in the $\mathcal{PD}$ index, with the specific value of the decrease.
For simplicity, let $r_{ij}:=[(\textbf{I}+\textbf{L})^{-1}]_{ij}$ and note that $r_{ij} \geq 0$ since $\textbf{I}+\textbf{L}$ is an M-matrix\footnote{
An M-matrix is defined as a matrix with non-positive off-diagonal entries and non-negative real eigenvalues.
}and, thus, is inverse-positive~\cite{plemmons1977m}. If the stubbornness of the $l$-th node in the network increases by $\epsilon$, then $\textbf{K}$ can be written as $\textbf{K}=\textbf{I}+\epsilon \textbf{e}_l \textbf{e}_l^T$ where $\textbf{e}_l \in \mathbb{R}^{n \times 1}$ is the $l$-th standard basis vector.
Using the Sherman–Morrison formula~\cite{sherman1950adjustment} 
\begin{equation*}
(\textbf{L}+\textbf{K})^{-1}=\bigg(\textbf{I}+\textbf{L}+\epsilon \textbf{e}_l \textbf{e}_l^T\bigg)^{-1}=(\textbf{I}+\textbf{L})^{-1}-\epsilon \frac{(\textbf{I}+\textbf{L})^{-1} \textbf{e}_l \textbf{e}_l^T (\textbf{I}+\textbf{L})^{-1}}{1+\epsilon  \textbf{e}_l^T (\textbf{I}+\textbf{L})^{-1}\textbf{e}_l}, 
\end{equation*}
where $\epsilon \textbf{e}_l^T (\textbf{I}+\textbf{L})^{-1} \textbf{e}_l=\epsilon r_{ll}$. So, 
\begin{equation*}
(\textbf{L}+\textbf{K})^{-1} \textbf{K}=(\textbf{I}+\textbf{L})^{-1}+\epsilon (\textbf{I}+\textbf{L})^{-1} \textbf{e}_l \textbf{e}_l^T
\end{equation*}
\begin{equation*}
-\epsilon \frac{(\textbf{I}+\textbf{L})^{-1} \textbf{e}_l \textbf{e}_l^T (\textbf{I}+\textbf{L})^{-1}}{1+\epsilon r_{ll}}-\epsilon^2 \frac{(\textbf{I}+\textbf{L})^{-1} \textbf{e}_l \textbf{e}_l^T (\textbf{I}+\textbf{L})^{-1}\textbf{e}_l \textbf{e}_l^T} {1+\epsilon r_{ll}}, 
\end{equation*}

which can be simplified as 
\begin{equation} \label{(L+K)^-1K*}
(\textbf{L}+\textbf{K})^{-1} \textbf{K}=(\textbf{I}+\textbf{L})^{-1}-\textbf{Z}+\frac{\epsilon}{1+\epsilon r_{ll}} (\textbf{I}+\textbf{L})^{-1} \textbf{e}_l \textbf{e}_l^T, 
\end{equation}
where $\textbf{Z}:=\epsilon \frac{(\textbf{I}+\textbf{L})^{-1} \textbf{e}_l \textbf{e}_l^T (\textbf{I}+\textbf{L})^{-1}}{1+\epsilon r_{ll}}$. Similar calculations show that 
\begin{equation}\label{K*(L+K)^-1}
\textbf{K}(\textbf{L}+\textbf{K})^{-1}=(\textbf{I}+\textbf{L})^{-1}-\textbf{Z}+\frac{\epsilon}{1+\epsilon r_{ll}} \textbf{e}_l \textbf{e}_l^T (\textbf{I}+\textbf{L})^{-1}. 
\end{equation}

\begin{theorem}
\label{thm:neutral}
Let \(\mathcal{PD} \) be the polarization-disagreement  in the FJ model and \(\widehat{\mathcal{PD}}\)
denotes the polarization-disagreement when the stubbornness of a specific node \( l \) with $s_l=0$ is increased by \(\epsilon\). Then,
\begin{equation*}
\widehat{\mathcal{PD}}=\mathcal{PD}-\frac{1}{n}  \langle s, \Vec{\textbf{1}}_{\textbf{K}} \rangle^2-\frac{2 \epsilon+\epsilon^2 r_{ll}}{(1+\epsilon r_{ll})^2}(\bar{z}_l^{\scriptscriptstyle{FJ}})^2.
\end{equation*}
\end{theorem}
\begin{proof}
Following Equations~\eqref{(L+K)^-1K*},~\eqref{K*(L+K)^-1}, and noting that $\textbf{e}_l \textbf{e}_l^T \bar{s}=\Vec{0}$ and also $\bar{s}^T \textbf{e}_l \textbf{e}_l^T=\Vec{0}$, the term $\bar{s} \textbf{P}^T \textbf{P} \bar{s}$ will be reduced to 
\begin{equation*}
\bar{s} \textbf{P}^T \textbf{P} \bar{s}=\bar{s} ^T (\textbf{I}+\textbf{L})^{-1} \bar{s} -2 \bar{s} ^T \textbf{Z} \bar{s}+ \bar{s} ^T \textbf{Z} (\textbf{I}+\textbf{L}) \textbf{Z} \bar{s},
\end{equation*}
\begin{equation*}
=\mathcal{PD} -2 \bar{s} ^T \textbf{Z} \bar{s}+ \frac{\epsilon r_{ll}}{1+\epsilon r_{ll}}\bar{s} ^T \textbf{Z} \bar{s}. 
\end{equation*}
As the expressed opinion in FJ model is given by $\bar{z}^{\scriptscriptstyle{FJ}}=(\textbf{I}+\textbf{L})^{-1}\bar{s}$ and  $\bar{s} ^T (\textbf{I}+\textbf{L})^{-1} \textbf{e}_l \textbf{e}_l^T (\textbf{I}+\textbf{L})^{-1} \bar{s}=(\bar{z}^{\scriptscriptstyle{FJ}})^T \textbf{e}_l \textbf{e}_l^T \bar{z}^{\scriptscriptstyle{FJ}}=(\bar{z}_l^{\scriptscriptstyle{FJ}})^2$ we will have 
\begin{equation*}
\bar{s} \textbf{P}^T \textbf{P} \bar{s}=\mathcal{PD}-\frac{2 \epsilon}{(1+\epsilon r_{ll})}(\bar{z}_l^{\scriptscriptstyle{FJ}})^2 +\frac{\epsilon^2 r_{ll}}{(1+\epsilon r_{ll})^2} (\bar{z}_l^{\scriptscriptstyle{FJ}})^2
\end{equation*}
which can be simplified as 
%\begin{equation}\label{s.TP.TPs}
$\bar{s} \textbf{P}^T \textbf{P} \bar{s}=\mathcal{PD}-\frac{2 \epsilon+\epsilon^2 r_{ll}}{(1+\epsilon r_{ll})^2}(\bar{z}_l^{\scriptscriptstyle{FJ}})^2. $
%\end{equation}
Applying Lemma~\ref{PD_hat} and substituting $\bar{s} \textbf{P}^T \textbf{P} \bar{s}$,  finishes the proof.
\end{proof}

\section{Experiments} \label{experiments}
In this section, we provide our experimental findings, which not only corroborate our theoretical results but also shed some lights on some other interesting questions.
% We conduct experiments to evaluate how accurately our theoretical bounds align with the behavior observed in real-world graph data and synthetic graph models.
\subsection{Setup}

\textbf{Synthetic Graphs.} We consider three classes of synthetic graphs: Erd\"{o}s–R\'{e}nyi (ER) random graph with parameter $p$ (the probability for each pair of nodes to form an edge, independently), Barabási-Albert (BA) graphs with parameter $m_{ba}$ (number of edges to attach from a new node to existing nodes), and the SBM consists of two blocks, where each intra-block edge appears with probability $p$ and each inter-block edge appears with probability $q$, independently (please see Section~\ref{homogeneous_section} for more details on SBM).  The number of nodes in all of these cases is $1000$.  \\
\textbf{Real-world Networks.} We also analyze 
Twitch-PT~\cite{rozemberczki2019multiscale} (1,912 nodes and 31,299 edges), Facebook~\cite{leskovec2012learning} (4,039 nodes  and 88,234 edges), Twitch-ES~\cite{rozemberczki2019multiscale} (4,648 nodes and 59,382 edges), LastFM Asia~\cite{feather} (7,624 nodes and 27,806 edges) ans also Twitch-DE~\cite{rozemberczki2019multiscale} (9,498 nodes and  153,138 edges) social networks. \\
\textbf{Innate Opinions.} We sampled innate opinions from the interval $[-1,1]$ using both uniform and Gaussian distributions.
\\
\textbf{Repetitions.} For synthetic graph cases, each experiment is repeated $100$ times, whereas for real-data sets, each experiment is repeated $50$ times and the average value is reported. \\
\textbf{Change Measure.} Our base case is the vanilla FJ model (where all nodes have stubbornness 1). Thus, when assessing changes in the $\mathcal{PD}$, we calculate the relative change to the $\mathcal{PD}$ in the FJ model, expressed as
%\begin{equation*}
$\frac{\mathcal{PD} - \mathcal{PD}^{\scriptscriptstyle{FJ}}}{\mathcal{PD}^{\scriptscriptstyle{FJ}}}$
%\end{equation*}
where $\mathcal{PD}$ represents the measured $\mathcal{PD}$ and $\mathcal{PD}^{\scriptscriptstyle{FJ}}$ denotes the $\mathcal{PD}$ in the FJ model.

\begin{table*}[h]\scriptsize
\centering
\caption{Effect of increasing stubbornness on the $\mathcal{PD}$ in six combinations of low, average, and high degree and neutral and non-neutral nodes.}  \label{table:Twitch_PT_FB}
\begin{tabular}{cccccccccccccccccccccccccccc}
\hline \hline
& & \multicolumn{2}{c}{Twitch-PT} & & \multicolumn{2}{c}{Facebook} & & \multicolumn{2}{c}{Twitch-ES} & & \multicolumn{2}{c}{LastFM Asia} & & \multicolumn{2}{c}{Twitch-DE}\\
\cline{3-4}\cline{6-7}\cline{9-10}\cline{12-13}\cline{15-16}
$Case$ & & $s = 0$ & $s \neq 0$ & & $s = 0$ & $s \neq 0$ & & $s = 0$ & $s \neq 0$ & & $s = 0$ & $s \neq 0$ & & $s = 0$ & $s \neq 0$ \\
\hline
Low & & -0.0\% (8\%) & 1.2\% (100\%) & & -0.0\% (6\%) & 1.8\% (100\%) & & -0.0\% (0\%)& 12\% (100\%) & & -0.0\% (0\%) & 2.0\% (100\%) & & -0.0\% (0\%) & 15\% (100\%)\\ 
Medium & & -0.0\% (20\%) & 1.4\% (100\%) & & -0.2\%  (0\%) & 0.8\% (100\%) & & -0.0\% (0\%) & 16\% (100\%) & & -0.0\% (0\%) & 1.6\% (100\%) & & -0.0\% (8\%) & 10\% (100\%)\\
High & & -0.1\% (3\%) & 1.9\% (100\%) & & -0.2\% (0\%) & 0.7\% (100\%) & & -0.0\% (0\%) & 15\% (100\%) & & -0.0\% (0\%) & 1.8\% (100\%) & & -0.0\% (0\%) & 18\% (100\%)\\
\hline \hline
\end{tabular}
\end{table*}

\subsection{Findings}

\textbf{I. Randomly Generated Setup}.
We first consider the setup where the innate opinion vectors are chosen based on uniform distribution. Initially, the stubbornness of all nodes are set to 1. Then, the stubbornness of a single node is increased to $10$. For different parameters in both ER and BA graphs, we calculated the average change in $\mathcal{PD}$ and the percentage of cases that resulted in an increase in $\mathcal{PD}$, as shown in Table~\ref{table:synthetic_graphs}.

The notation $2.8\% (98.4\%)$, for example, indicates that the average change is $2.8\%$, and for $98.4\%$ of the cases (nodes), the $\mathcal{PD}$ increases. As one might expect for most nodes, an increase in stubbornness is positively correlated with increase in $\mathcal{PD}$. In particular, in the ER graph for $p=0.5$, where one of all graphs on $n$ nodes is chosen uniformly at random, the increase occurs for $98.4\%$ of the nodes. Furthermore, we observe that in ER graph as $p$ approaches 1 (complete graph), this relation becomes stronger, that is, for a larger percentage of nodes an increase in stubbornness exacerbates the $\mathcal{PD}$. Similar results were observed for the BA graph, which resembles real-world networks. For the real-world data, our results are as follows: Twitch-PT graph $0.7\% \, (85\%)$, Facebook $0.8\% \, (98\%)$, Twitch-ES $0.4\% \, (100\%)$, LastFM $0.1\% \, (97\%)$, and Twitch-DE $0.0\% \, (95\%)$.
\begin{center}
\textit{An increase in stubbornness ``usually'' causes an increase in $\mathcal{PD}$.}
\end{center}

\begin{table}[h]
\centering
\caption{Average Change in $\mathcal{PD}$ and Percentage of Positive Samples (P.S) for Different Parameters in ER and BA Graphs.}
\label{table:synthetic_graphs}
\begin{tabular}{cccccccccc}
\toprule
\textbf{$p$} & \textbf{ER: Change\% (P.S)} & & \textbf{$m_{ba}$} & \textbf{BA: Change\% (P.S)} \\
\midrule
0.05 & 1.6\% (91.2\%) & & 1 & 0.0\% (84.0\%) \\
0.25 & 2.6\% (97.7\%) & & 2 & 0.1\% (89.0\%) \\
0.50 & 2.8\% (98.4\%) & & 3 & 0.1\% (93.0\%) \\
0.75 & 2.8\% (99.4\%) & & 4 & 0.1\% (92.0\%) \\
0.95 & 3.0\% (99.5\%) & & 5 & 0.2\% (97.0\%) \\
1.00 & 3.0\% (100\%) & & 6 & 0.2\% (94.0\%) \\
\bottomrule
\end{tabular}
\end{table}

\noindent \textbf{II. Magnitude of Change.} To get a better understanding of the magnitude of change in $\mathcal{PD}$ %\ahad{We are using both $\mathcal{PD}$ and the $\mathcal{PD}$ now. I'm fine with both, but we should make sure to be consistent everywhere. \mohammad{red}{What is the difference between $\mathcal{PD}$ and $\mathcal{PD}$? they are the same. We used mathcal\{PD\} for the PD}}
as a result of increasing stubbornness, we provide the results in Figure~\ref{PD_versus_alphe_homogeneous}. We consider the homogeneous setup and vary the stubbornness factor $\alpha$ and track the changes in $\mathcal{PD}$ with respect to the vanilla FJ model ($\alpha=1$). We considered all considered networks and the ER and BA graphs (where the parameters for these graphs were chosen to make their number of nodes and average degree comparable to the real-world graphs). A uniform distribution was used for the choice of innate opinions.

\begin{figure}[h]
\includegraphics[width=6cm, height=5cm]{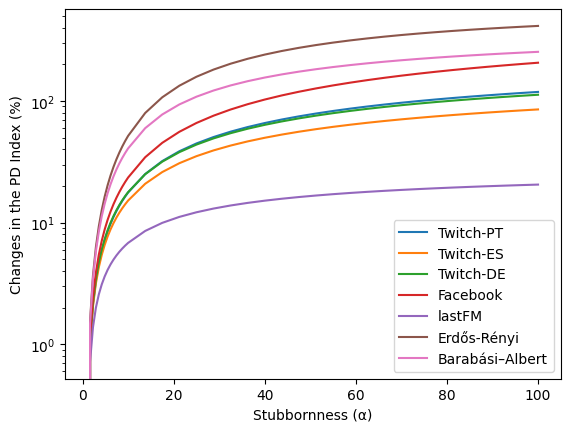}
\caption{The change in the $\mathcal{PD}$ (in logarithmic scale) with variation in (homogeneous) stubbornness across different networks. %\ahad{Shouldn't all curves start at $10^0$? Some of them seem to start higher.}
} \label{PD_versus_alphe_homogeneous}
\end{figure}
Aligned with our theoretical findings (Theorem~\ref{P_PD_K_alpha_I_theorem}), higher homogeneous stubbornness in the network correlates with an increased $\mathcal{PD}$. In fact, large values of $\alpha$ can almost double the value of $\mathcal{PD}$. However, as $\alpha$ increases, its impact starts to vanish.
The explanation for this is that as the stubbornness of a node increases, it gives more weight to its innate opinion until it eventually almost disregards its neighbors' opinion. Increasing the stubbornness beyond that point doesn't have much of an impact. Particularly, once the stubbornness of a node passes its degree, it already gives more weights to its own opinion than its neighbors. This is in line with the observation that the increase in $\mathcal{PD}$ is most dramatic when the stubbornness factor is below $20$-$30$, which is comparable to the average degree in these graphs (to be more precise, the average degree for Twitch-PT is $32$, Twitch-ES $25$, Twitch-DE $32$, Facebook $43$, LastFM $7$, ER $25$ and BA $23$).
% This is because when the stubbornness of a node is larger than its degree, it already gives more weight to its innate opinion than its neighbors' opinions combined.

\begin{center}
\textit{Increasing stubbornness increases $\mathcal{PD}$ ``significantly'' and there is a smooth phase transition around the average degree.}
\end{center}

\noindent \textbf{III. The Governing Parameters.}
We consider nodes of low degree, average degree, and high degree. For each of these, we consider nodes with an innate opinion close to 0 (neutral nodes) and nodes with an innate opinion smaller/larger than 0. This gives of 6 possible cases. We study the change in $\mathcal{PD}$ for all these cases and the results are reported in Table~\ref{table:Twitch_PT_FB}.

We used a uniform distribution for innate opinions. Then, we increased the stubbornness of a selected subset of nodes ($1\%$ of the total) to a value of $10$. The number of instances where $\mathcal{PD}$ either increased (positive samples) or decreased (negative samples) were calculated. We, first, observe that in most cases, increasing the stubbornness of a node with a non-zero innate opinion tends to increase $\mathcal{PD}$. Conversely, increasing the stubbornness of neutral nodes tends to decrease $\mathcal{PD}$. (This is inline with our theoretical results from Theorem~\ref{thm:neutral}.) The decrease in the second scenario is relatively small, while the increase in the first scenario (increasing the stubbornness of a node with a non-zero innate opinion) can lead to significant increase in the $\mathcal{PD}$.
% More precisely, in all cases reported in these tables, the change in $\mathcal{PD}$ when $s = 0$ is less than $1\%$, while changes can reach up to $15\%$ in the Twitch-ES case.

We also observe that the degree of a node, surprisingly, does not seem to play a significant role in the magnitude of change. In all cases, the magnitude of change remains consistent across both low degree and high degree nodes, with no significant difference observed.

\begin{center}
\textit{While being neutral or not determines whether a node's stubbornness increases or decreases $\mathcal{PD}$, its degree doesn't seem to have an impact on the magnitude of the change.}
\end{center}

\noindent \textbf{IV. Can Stubborn Extreme Nodes Reduce Polarization?} Surprisingly, the answer to this question is yes. So far, we demonstrated, both theoretically and experimentally, increasing the stubbornness of neutral nodes results in lower $\mathcal{PD}$. This is surprising by itself, but even more surprisingly, we observe that the carefully chosen extreme nodes (nodes with innate opinion far from 0) can produce a similar outcome.
% \textit{For the final part of the experiment, we will demonstrate that increasing the stubbornness of neutral nodes is not the only scenario that helps reduce the $\mathcal{PD}$}. In the following part, we specify a setup that help the $\mathcal{PD}$ reduce in the case of increasing the stubbornness of some nodes.
For this purpose, we consider the SBM with two blocks where $p=0.30$ and $q$ is variable. The nodes' opinions are chosen based on a Gaussian distribution between $-1$ and $+1$ for both groups. The opinions in the first and second group are biased towards $-1/2$ and $+1/2$, respectively. This is consistent with real-world scenarios, where individuals form bubbles with other people of similar opinion.

We selected the nodes in each block holding the most opposing opinions and increased their stubbornness to $+10$. The average changes in the $\mathcal{PD}$ index are plotted in Figure~\ref{P_D_PD}. The main observation that can be drawn is that for small values of $q$, the $\mathcal{PD}$ decreases compared to the FJ model. This is intuitive because there are fewer interactions between different bubbles, and as a result the blocks converge towards $-1/2$ and $+1/2$, respectively. In that case, stubborn nodes of opposing opinion in each side push the average opinion towards 0, which helps to reduce $\mathcal{PD}$. However, as interactions between nodes increase ($q$ becomes larger), this change has the opposite effect. This is because as $q$ grows, the bubble structure starts vanishing. In that case, the nodes will already converge to a value close to 0 and adding stubborn nodes can cause an increase in $\mathcal{PD}$, rather than decrease.

\begin{figure}[H]
\includegraphics[width=7cm, height=5cm]{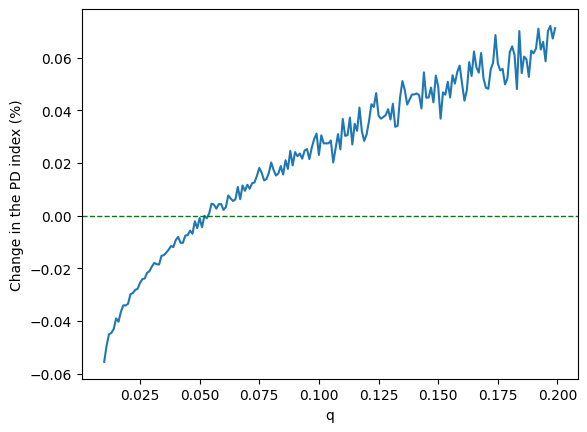}
\caption{The average change in the $\mathcal{PD}$ in the SBM with two positive/negative bubbles when the stubbornness of nodes with opposing opinions in each bubble is increased.}
\label{P_D_PD}
\end{figure}

\begin{center}
\textit{Increasing the stubbornness of nodes with an opposing opinion in ``bubbles'' can reduce $\mathcal{PD}$.}
\end{center}

\section{Conclusion} \label{conclusion}

We studied the impact of users' stubbornness on the polarization-disagreement ($\mathcal{PD}$), in the framework of Friedkin-Johnsen opinion formation model. We investigated the characterization of the circumstances where an increase in stubbornness causes an increase/decrease in $\mathcal{PD}$. Based on our theoretical and experimental results, we observed that in most scenarios stubbornness and $\mathcal{PD}$ are positively correlated. However, interestingly, we discovered two types of nodes whose increase in stubbornness can reduce $\mathcal{PD}$: neutral nodes (nodes whose innate opinion is close to the mean opinion) and nodes who have an opposing opinion to most nodes in their bubble.

There are several potential avenues for future research. Our experimental results strongly suggest that in complete graphs, increasing the stubbornness coefficient of a node will lead to an increase in $\mathcal{PD}$, regardless of the innate opinion vector. Can this be validated? What are other graph structures where increasing the stubbornness of some nodes will always increase the total $\mathcal{PD}$? Another intriguing question is: given a fixed budget \(k\), how should we choose \(k\) nodes in a network and adjust their stubbornness to achieve the maximum reduction in the $\mathcal{PD}$?

\bibliographystyle{ACM-Reference-Format}
% \bibliography{References}

\begin{thebibliography}{44}

%%% the Web address.

\ifx \showCODEN    \undefined \def \showCODEN     #1{\unskip}     \fi
\ifx \showDOI      \undefined \def \showDOI       #1{#1}\fi
\ifx \showISBNx    \undefined \def \showISBNx     #1{\unskip}     \fi
\ifx \showISBNxiii \undefined \def \showISBNxiii  #1{\unskip}     \fi
\ifx \showISSN     \undefined \def \showISSN      #1{\unskip}     \fi
\ifx \showLCCN     \undefined \def \showLCCN      #1{\unskip}     \fi
\ifx \shownote     \undefined \def \shownote      #1{#1}          \fi
\ifx \showarticletitle \undefined \def \showarticletitle #1{#1}   \fi
\ifx \showURL      \undefined \def \showURL       {\relax}        \fi
% invisible to TeX
\providecommand\bibfield[2]{#2}
\providecommand\bibinfo[2]{#2}
\providecommand\natexlab[1]{#1}
\providecommand\showeprint[2][]{arXiv:#2}

\bibitem[Abebe et~al\mbox{.}(2021)]%
        {abebe2021opinion}
\bibfield{author}{\bibinfo{person}{Rediet Abebe}, \bibinfo{person}{T-H~HUBERT Chan}, \bibinfo{person}{Jon Kleinberg}, \bibinfo{person}{Zhibin Liang}, \bibinfo{person}{David Parkes}, \bibinfo{person}{Mauro Sozio}, {and} \bibinfo{person}{Charalampos~E Tsourakakis}.} \bibinfo{year}{2021}\natexlab{}.
\newblock \showarticletitle{Opinion dynamics optimization by varying susceptibility to persuasion via non-convex local search}.
\newblock \bibinfo{journal}{\emph{ACM Transactions on Knowledge Discovery from Data (TKDD)}} \bibinfo{volume}{16}, \bibinfo{number}{2} (\bibinfo{year}{2021}), \bibinfo{pages}{1--34}.
\newblock


\bibitem[Abebe et~al\mbox{.}(2018)]%
        {abebe2018opinion}
\bibfield{author}{\bibinfo{person}{Rediet Abebe}, \bibinfo{person}{Jon Kleinberg}, \bibinfo{person}{David Parkes}, {and} \bibinfo{person}{Charalampos~E Tsourakakis}.} \bibinfo{year}{2018}\natexlab{}.
\newblock \showarticletitle{Opinion dynamics with varying susceptibility to persuasion}. In \bibinfo{booktitle}{\emph{Proceedings of the 24th ACM SIGKDD International Conference on Knowledge Discovery \& Data Mining}}. \bibinfo{pages}{1089--1098}.
\newblock


\bibitem[Acemoglu and Ozdaglar(2011)]%
        {acemoglu2011opinion}
\bibfield{author}{\bibinfo{person}{Daron Acemoglu} {and} \bibinfo{person}{Asuman Ozdaglar}.} \bibinfo{year}{2011}\natexlab{}.
\newblock \showarticletitle{Opinion dynamics and learning in social networks}.
\newblock \bibinfo{journal}{\emph{Dynamic Games and Applications}}  \bibinfo{volume}{1} (\bibinfo{year}{2011}), \bibinfo{pages}{3--49}.
\newblock


\bibitem[Bhalla et~al\mbox{.}(2023)]%
        {bhalla2023local}
\bibfield{author}{\bibinfo{person}{Nikita Bhalla}, \bibinfo{person}{Adam Lechowicz}, {and} \bibinfo{person}{Cameron Musco}.} \bibinfo{year}{2023}\natexlab{}.
\newblock \showarticletitle{Local edge dynamics and opinion polarization}. In \bibinfo{booktitle}{\emph{Proceedings of the Sixteenth ACM International Conference on Web Search and Data Mining}}. \bibinfo{pages}{6--14}.
\newblock


\bibitem[Bindel et~al\mbox{.}(2015)]%
        {bindel2015bad}
\bibfield{author}{\bibinfo{person}{David Bindel}, \bibinfo{person}{Jon Kleinberg}, {and} \bibinfo{person}{Sigal Oren}.} \bibinfo{year}{2015}\natexlab{}.
\newblock \showarticletitle{How bad is forming your own opinion?}
\newblock \bibinfo{journal}{\emph{Games and Economic Behavior}}  \bibinfo{volume}{92} (\bibinfo{year}{2015}), \bibinfo{pages}{248--265}.
\newblock


\bibitem[Chitra and Musco(2019)]%
        {chitra2019understanding}
\bibfield{author}{\bibinfo{person}{Uthsav Chitra} {and} \bibinfo{person}{Christopher Musco}.} \bibinfo{year}{2019}\natexlab{}.
\newblock \showarticletitle{Understanding filter bubbles and polarization in social networks}.
\newblock \bibinfo{journal}{\emph{arXiv preprint arXiv:1906.08772}} (\bibinfo{year}{2019}).
\newblock


\bibitem[Chitra and Musco(2020)]%
        {chitra2020analyzing}
\bibfield{author}{\bibinfo{person}{Uthsav Chitra} {and} \bibinfo{person}{Christopher Musco}.} \bibinfo{year}{2020}\natexlab{}.
\newblock \showarticletitle{Analyzing the impact of filter bubbles on social network polarization}. In \bibinfo{booktitle}{\emph{Proceedings of the 13th International Conference on Web Search and Data Mining}}. \bibinfo{pages}{115--123}.
\newblock


\bibitem[Das et~al\mbox{.}(2014)]%
        {das2014modeling}
\bibfield{author}{\bibinfo{person}{Abhimanyu Das}, \bibinfo{person}{Sreenivas Gollapudi}, {and} \bibinfo{person}{Kamesh Munagala}.} \bibinfo{year}{2014}\natexlab{}.
\newblock \showarticletitle{Modeling opinion dynamics in social networks}. In \bibinfo{booktitle}{\emph{Proceedings of the 7th ACM International Conference on Web search and Data Mining}}. \bibinfo{pages}{403--412}.
\newblock


\bibitem[Deffuant et~al\mbox{.}(2000)]%
        {deffuant2000mixing}
\bibfield{author}{\bibinfo{person}{Guillaume Deffuant}, \bibinfo{person}{David Neau}, \bibinfo{person}{Frederic Amblard}, {and} \bibinfo{person}{G{\'e}rard Weisbuch}.} \bibinfo{year}{2000}\natexlab{}.
\newblock \showarticletitle{Mixing beliefs among interacting agents}.
\newblock \bibinfo{journal}{\emph{Advances in Complex Systems}} \bibinfo{volume}{3}, \bibinfo{number}{01n04} (\bibinfo{year}{2000}), \bibinfo{pages}{87--98}.
\newblock


\bibitem[DeGroot(1974)]%
        {degroot1974reaching}
\bibfield{author}{\bibinfo{person}{Morris~H DeGroot}.} \bibinfo{year}{1974}\natexlab{}.
\newblock \showarticletitle{Reaching a consensus}.
\newblock \bibinfo{journal}{\emph{Journal of the American Statistical association}} \bibinfo{volume}{69}, \bibinfo{number}{345} (\bibinfo{year}{1974}), \bibinfo{pages}{118--121}.
\newblock


\bibitem[DellaPosta(2020)]%
        {dellaposta2020pluralistic}
\bibfield{author}{\bibinfo{person}{Daniel DellaPosta}.} \bibinfo{year}{2020}\natexlab{}.
\newblock \showarticletitle{Pluralistic collapse: The “oil spill” model of mass opinion polarization}.
\newblock \bibinfo{journal}{\emph{American Sociological Review}} \bibinfo{volume}{85}, \bibinfo{number}{3} (\bibinfo{year}{2020}), \bibinfo{pages}{507--536}.
\newblock


\bibitem[Friedkin and Johnsen(1990)]%
        {friedkin1990social}
\bibfield{author}{\bibinfo{person}{Noah~E Friedkin} {and} \bibinfo{person}{Eugene~C Johnsen}.} \bibinfo{year}{1990}\natexlab{}.
\newblock \showarticletitle{Social influence and opinions}.
\newblock \bibinfo{journal}{\emph{Journal of Mathematical Sociology}} \bibinfo{volume}{15}, \bibinfo{number}{3-4} (\bibinfo{year}{1990}), \bibinfo{pages}{193--206}.
\newblock


\bibitem[G{\"a}rtner and Zehmakan(2018)]%
        {gartner2018majority}
\bibfield{author}{\bibinfo{person}{Bernd G{\"a}rtner} {and} \bibinfo{person}{Ahad~N Zehmakan}.} \bibinfo{year}{2018}\natexlab{}.
\newblock \showarticletitle{Majority model on random regular graphs}. In \bibinfo{booktitle}{\emph{LATIN 2018: Theoretical Informatics: 13th Latin American Symposium, Buenos Aires, Argentina, April 16-19, 2018, Proceedings 13}}. Springer, \bibinfo{pages}{572--583}.
\newblock


\bibitem[G{\"a}rtner and Zehmakan(2020)]%
        {gartner2020threshold}
\bibfield{author}{\bibinfo{person}{Bernd G{\"a}rtner} {and} \bibinfo{person}{Ahad~N Zehmakan}.} \bibinfo{year}{2020}\natexlab{}.
\newblock \showarticletitle{Threshold behavior of democratic opinion dynamics}.
\newblock \bibinfo{journal}{\emph{Journal of Statistical Physics}}  \bibinfo{volume}{178} (\bibinfo{year}{2020}), \bibinfo{pages}{1442--1466}.
\newblock


\bibitem[Golub and Van~Loan(2013)]%
        {golub2013matrix}
\bibfield{author}{\bibinfo{person}{Gene~H Golub} {and} \bibinfo{person}{Charles~F Van~Loan}.} \bibinfo{year}{2013}\natexlab{}.
\newblock \bibinfo{booktitle}{\emph{Matrix Computations}}.
\newblock \bibinfo{publisher}{JHU Press}.
\newblock


\bibitem[Haddadan et~al\mbox{.}(2021)]%
        {haddadan2021repbublik}
\bibfield{author}{\bibinfo{person}{Shahrzad Haddadan}, \bibinfo{person}{Cristina Menghini}, \bibinfo{person}{Matteo Riondato}, {and} \bibinfo{person}{Eli Upfal}.} \bibinfo{year}{2021}\natexlab{}.
\newblock \showarticletitle{Repbublik: Reducing polarized bubble radius with link insertions}. In \bibinfo{booktitle}{\emph{Proceedings of the 14th ACM International Conference on Web Search and Data Mining}}. \bibinfo{pages}{139--147}.
\newblock


\bibitem[Hart and Nisbet(2012)]%
        {hart2012boomerang}
\bibfield{author}{\bibinfo{person}{P~Sol Hart} {and} \bibinfo{person}{Erik~C Nisbet}.} \bibinfo{year}{2012}\natexlab{}.
\newblock \showarticletitle{Boomerang effects in science communication: How motivated reasoning and identity cues amplify opinion polarization about climate mitigation policies}.
\newblock \bibinfo{journal}{\emph{Communication research}} \bibinfo{volume}{39}, \bibinfo{number}{6} (\bibinfo{year}{2012}), \bibinfo{pages}{701--723}.
\newblock


\bibitem[Horn and Johnson(2012)]%
        {horn2012matrix}
\bibfield{author}{\bibinfo{person}{Roger~A Horn} {and} \bibinfo{person}{Charles~R Johnson}.} \bibinfo{year}{2012}\natexlab{}.
\newblock \bibinfo{booktitle}{\emph{Matrix Analysis}}.
\newblock \bibinfo{publisher}{Cambridge University Press}.
\newblock


\bibitem[Huang et~al\mbox{.}(2022)]%
        {huang2022pole}
\bibfield{author}{\bibinfo{person}{Zexi Huang}, \bibinfo{person}{Arlei Silva}, {and} \bibinfo{person}{Ambuj Singh}.} \bibinfo{year}{2022}\natexlab{}.
\newblock \showarticletitle{Pole: Polarized embedding for signed networks}. In \bibinfo{booktitle}{\emph{Proceedings of the Fifteenth ACM International Conference on Web Search and Data Mining}}. \bibinfo{pages}{390--400}.
\newblock


\bibitem[Hunter and Zaman(2022)]%
        {hunter2022optimizing}
\bibfield{author}{\bibinfo{person}{David~Scott Hunter} {and} \bibinfo{person}{Tauhid Zaman}.} \bibinfo{year}{2022}\natexlab{}.
\newblock \showarticletitle{Optimizing opinions with stubborn agents}.
\newblock \bibinfo{journal}{\emph{Operations Research}} \bibinfo{volume}{70}, \bibinfo{number}{4} (\bibinfo{year}{2022}), \bibinfo{pages}{2119--2137}.
\newblock


\bibitem[Leskovec and Mcauley(2012)]%
        {leskovec2012learning}
\bibfield{author}{\bibinfo{person}{Jure Leskovec} {and} \bibinfo{person}{Julian Mcauley}.} \bibinfo{year}{2012}\natexlab{}.
\newblock \showarticletitle{Learning to discover social circles in ego networks}.
\newblock \bibinfo{journal}{\emph{Advances in Neural Information Processing Systems}}  \bibinfo{volume}{25} (\bibinfo{year}{2012}).
\newblock


\bibitem[Musco et~al\mbox{.}(2018)]%
        {musco2018minimizing}
\bibfield{author}{\bibinfo{person}{Cameron Musco}, \bibinfo{person}{Christopher Musco}, {and} \bibinfo{person}{Charalampos~E Tsourakakis}.} \bibinfo{year}{2018}\natexlab{}.
\newblock \showarticletitle{Minimizing polarization and disagreement in social networks}. In \bibinfo{booktitle}{\emph{Proceedings of the 2018 World Wide Web Conference}}. \bibinfo{pages}{369--378}.
\newblock


\bibitem[N~Zehmakan and Galam(2020)]%
        {n2020rumor}
\bibfield{author}{\bibinfo{person}{Ahad N~Zehmakan} {and} \bibinfo{person}{Serge Galam}.} \bibinfo{year}{2020}\natexlab{}.
\newblock \showarticletitle{Rumor spreading: A trigger for proliferation or fading away}.
\newblock \bibinfo{journal}{\emph{Chaos: An Interdisciplinary Journal of Nonlinear Science}} \bibinfo{volume}{30}, \bibinfo{number}{7} (\bibinfo{year}{2020}).
\newblock


\bibitem[Neumann et~al\mbox{.}(2024)]%
        {neumann2024sublinear}
\bibfield{author}{\bibinfo{person}{Stefan Neumann}, \bibinfo{person}{Yinhao Dong}, {and} \bibinfo{person}{Pan Peng}.} \bibinfo{year}{2024}\natexlab{}.
\newblock \showarticletitle{Sublinear-Time Opinion Estimation in the Friedkin--Johnsen Model}. In \bibinfo{booktitle}{\emph{Proceedings of the ACM on Web Conference 2024}}. \bibinfo{pages}{2563--2571}.
\newblock


\bibitem[Noorazar(2020)]%
        {noorazar2020recent}
\bibfield{author}{\bibinfo{person}{Hossein Noorazar}.} \bibinfo{year}{2020}\natexlab{}.
\newblock \showarticletitle{Recent advances in opinion propagation dynamics: A 2020 survey}.
\newblock \bibinfo{journal}{\emph{The European Physical Journal Plus}}  \bibinfo{volume}{135} (\bibinfo{year}{2020}), \bibinfo{pages}{1--20}.
\newblock


\bibitem[Out et~al\mbox{.}(2024)]%
        {out2024impact}
\bibfield{author}{\bibinfo{person}{Charlotte Out}, \bibinfo{person}{Sijing Tu}, \bibinfo{person}{Stefan Neumann}, {and} \bibinfo{person}{Ahad~N Zehmakan}.} \bibinfo{year}{2024}\natexlab{}.
\newblock \showarticletitle{The Impact of External Sources on the Friedkin--Johnsen Model}. In \bibinfo{booktitle}{\emph{Proceedings of the 33rd ACM International Conference on Information and Knowledge Management}}. \bibinfo{pages}{1815--1824}.
\newblock


\bibitem[Out and Zehmakan(2021)]%
        {out2021majority}
\bibfield{author}{\bibinfo{person}{Charlotte Out} {and} \bibinfo{person}{Ahad~N Zehmakan}.} \bibinfo{year}{2021}\natexlab{}.
\newblock \showarticletitle{Majority vote in social networks: Make random friends or be stubborn to overpower elites}.
\newblock \bibinfo{journal}{\emph{The 30th International Joint Conference on Artificial Intelligence}} (\bibinfo{year}{2021}).
\newblock


\bibitem[Plemmons(1977)]%
        {plemmons1977m}
\bibfield{author}{\bibinfo{person}{Robert~J Plemmons}.} \bibinfo{year}{1977}\natexlab{}.
\newblock \showarticletitle{M-matrix characterizations. I—nonsingular M-matrices}.
\newblock \bibinfo{journal}{\emph{Linear Algebra Appl.}} \bibinfo{volume}{18}, \bibinfo{number}{2} (\bibinfo{year}{1977}), \bibinfo{pages}{175--188}.
\newblock


\bibitem[R{\'a}cz and Rigobon(2023)]%
        {racz2023towards}
\bibfield{author}{\bibinfo{person}{Miklos~Z R{\'a}cz} {and} \bibinfo{person}{Daniel~E Rigobon}.} \bibinfo{year}{2023}\natexlab{}.
\newblock \showarticletitle{Towards consensus: Reducing polarization by perturbing social networks}.
\newblock \bibinfo{journal}{\emph{IEEE Transactions on Network Science and Engineering}} (\bibinfo{year}{2023}).
\newblock


\bibitem[Rastegarpanah et~al\mbox{.}(2019)]%
        {rastegarpanah2019fighting}
\bibfield{author}{\bibinfo{person}{Bashir Rastegarpanah}, \bibinfo{person}{Krishna~P Gummadi}, {and} \bibinfo{person}{Mark Crovella}.} \bibinfo{year}{2019}\natexlab{}.
\newblock \showarticletitle{Fighting fire with fire: Using antidote data to improve polarization and fairness of recommender systems}. In \bibinfo{booktitle}{\emph{Proceedings of the twelfth ACM international Conference on Web Search and Data Mining}}. \bibinfo{pages}{231--239}.
\newblock


\bibitem[Rozemberczki et~al\mbox{.}(2021)]%
        {rozemberczki2019multiscale}
\bibfield{author}{\bibinfo{person}{Benedek Rozemberczki}, \bibinfo{person}{Carl Allen}, {and} \bibinfo{person}{Rik Sarkar}.} \bibinfo{year}{2021}\natexlab{}.
\newblock \showarticletitle{Multi-scale attributed node embedding}.
\newblock \bibinfo{journal}{\emph{Journal of Complex Networks}} \bibinfo{volume}{9}, \bibinfo{number}{2} (\bibinfo{year}{2021}), \bibinfo{pages}{cnab014}.
\newblock


\bibitem[Rozemberczki and Sarkar(2020)]%
        {feather}
\bibfield{author}{\bibinfo{person}{Benedek Rozemberczki} {and} \bibinfo{person}{Rik Sarkar}.} \bibinfo{year}{2020}\natexlab{}.
\newblock \showarticletitle{{Characteristic Functions on Graphs: Birds of a Feather, from Statistical Descriptors to Parametric Models}}. In \bibinfo{booktitle}{\emph{Proceedings of the 29th ACM International Conference on Information and Knowledge Management (CIKM '20)}}. ACM, \bibinfo{pages}{1325–1334}.
\newblock


\bibitem[Sherman and Morrison(1950)]%
        {sherman1950adjustment}
\bibfield{author}{\bibinfo{person}{Jack Sherman} {and} \bibinfo{person}{Winifred~J Morrison}.} \bibinfo{year}{1950}\natexlab{}.
\newblock \showarticletitle{Adjustment of an inverse matrix corresponding to a change in one element of a given matrix}.
\newblock \bibinfo{journal}{\emph{The Annals of Mathematical Statistics}} \bibinfo{volume}{21}, \bibinfo{number}{1} (\bibinfo{year}{1950}), \bibinfo{pages}{124--127}.
\newblock


\bibitem[Steinberg and Monahan(2007)]%
        {steinberg2007age}
\bibfield{author}{\bibinfo{person}{Laurence Steinberg} {and} \bibinfo{person}{Kathryn~C Monahan}.} \bibinfo{year}{2007}\natexlab{}.
\newblock \showarticletitle{Age differences in resistance to peer influence.}
\newblock \bibinfo{journal}{\emph{Developmental Psychology}} \bibinfo{volume}{43}, \bibinfo{number}{6} (\bibinfo{year}{2007}), \bibinfo{pages}{1531}.
\newblock


\bibitem[Sun et~al\mbox{.}(2023a)]%
        {sun2023all}
\bibfield{author}{\bibinfo{person}{Xiangguo Sun}, \bibinfo{person}{Hong Cheng}, \bibinfo{person}{Jia Li}, \bibinfo{person}{Bo Liu}, {and} \bibinfo{person}{Jihong Guan}.} \bibinfo{year}{2023}\natexlab{a}.
\newblock \showarticletitle{All in one: Multi-task prompting for graph neural networks}. In \bibinfo{booktitle}{\emph{Proceedings of the 29th ACM SIGKDD Conference on Knowledge Discovery and Data Mining}}. \bibinfo{pages}{2120--2131}.
\newblock


\bibitem[Sun et~al\mbox{.}(2023b)]%
        {sun2023self}
\bibfield{author}{\bibinfo{person}{Xiangguo Sun}, \bibinfo{person}{Hong Cheng}, \bibinfo{person}{Bo Liu}, \bibinfo{person}{Jia Li}, \bibinfo{person}{Hongyang Chen}, \bibinfo{person}{Guandong Xu}, {and} \bibinfo{person}{Hongzhi Yin}.} \bibinfo{year}{2023}\natexlab{b}.
\newblock \showarticletitle{Self-supervised hypergraph representation learning for sociological analysis}.
\newblock \bibinfo{journal}{\emph{IEEE Transactions on Knowledge and Data Engineering}} \bibinfo{volume}{35}, \bibinfo{number}{11} (\bibinfo{year}{2023}), \bibinfo{pages}{11860--11871}.
\newblock


\bibitem[Wall et~al\mbox{.}(1993)]%
        {wall1993susceptibility}
\bibfield{author}{\bibinfo{person}{Julie~A Wall}, \bibinfo{person}{Thomas~G Power}, {and} \bibinfo{person}{Consuelo Arbona}.} \bibinfo{year}{1993}\natexlab{}.
\newblock \showarticletitle{Susceptibility to antisocial peer pressure and its relation to acculturation in Mexican-American adolescents}.
\newblock \bibinfo{journal}{\emph{Journal of Adolescent Research}} \bibinfo{volume}{8}, \bibinfo{number}{4} (\bibinfo{year}{1993}), \bibinfo{pages}{403--418}.
\newblock


\bibitem[Wang and Kleinberg(2024)]%
        {wang2024relationship}
\bibfield{author}{\bibinfo{person}{Yanbang Wang} {and} \bibinfo{person}{Jon Kleinberg}.} \bibinfo{year}{2024}\natexlab{}.
\newblock \showarticletitle{On the relationship between relevance and conflict in online social link recommendations}.
\newblock \bibinfo{journal}{\emph{Advances in Neural Information Processing Systems}}  \bibinfo{volume}{36} (\bibinfo{year}{2024}).
\newblock


\bibitem[Xu et~al\mbox{.}(2022)]%
        {xu2022effects}
\bibfield{author}{\bibinfo{person}{Wanyue Xu}, \bibinfo{person}{Liwang Zhu}, \bibinfo{person}{Jiale Guan}, \bibinfo{person}{Zuobai Zhang}, {and} \bibinfo{person}{Zhongzhi Zhang}.} \bibinfo{year}{2022}\natexlab{}.
\newblock \showarticletitle{Effects of stubbornness on opinion dynamics}. In \bibinfo{booktitle}{\emph{Proceedings of the 31st ACM International Conference on Information \& Knowledge Management}}. \bibinfo{pages}{2321--2330}.
\newblock


\bibitem[Zehmakan(2019)]%
        {zehmakan2019spread}
\bibfield{author}{\bibinfo{person}{Abdolahad~N Zehmakan}.} \bibinfo{year}{2019}\natexlab{}.
\newblock \emph{\bibinfo{title}{On the spread of information through graphs}}.
\newblock \bibinfo{thesistype}{Ph.\,D. Dissertation}. \bibinfo{school}{ETH Zurich}.
\newblock


\bibitem[Zehmakan(2021)]%
        {zehmakan2021majority}
\bibfield{author}{\bibinfo{person}{Ahad~N Zehmakan}.} \bibinfo{year}{2021}\natexlab{}.
\newblock \showarticletitle{Majority opinion diffusion in social networks: An adversarial approach}. In \bibinfo{booktitle}{\emph{Proceedings of the AAAI Conference on Artificial Intelligence}}, Vol.~\bibinfo{volume}{35}. \bibinfo{pages}{5611--5619}.
\newblock


\bibitem[Zehmakan(2023)]%
        {zehmakan2023random}
\bibfield{author}{\bibinfo{person}{Ahad~N Zehmakan}.} \bibinfo{year}{2023}\natexlab{}.
\newblock \showarticletitle{Random Majority Opinion Diffusion: Stabilization Time, Absorbing States, and Influential Nodes}. In \bibinfo{booktitle}{\emph{Proceedings of the 2023 International Conference on Autonomous Agents and Multiagent Systems}}. \bibinfo{pages}{2179--2187}.
\newblock


\bibitem[Zehmakan et~al\mbox{.}(2024)]%
        {zehmakan2024viral}
\bibfield{author}{\bibinfo{person}{Ahad~N Zehmakan}, \bibinfo{person}{Xiaotian Zhou}, {and} \bibinfo{person}{Zhongzhi Zhang}.} \bibinfo{year}{2024}\natexlab{}.
\newblock \showarticletitle{Viral Marketing in Social Networks with Competing Products}. In \bibinfo{booktitle}{\emph{Proceedings of the 23rd International Conference on Autonomous Agents and Multiagent Systems}}. \bibinfo{pages}{2047--2056}.
\newblock


\bibitem[Zhu et~al\mbox{.}(2021)]%
        {zhu2021minimizing}
\bibfield{author}{\bibinfo{person}{Liwang Zhu}, \bibinfo{person}{Qi Bao}, {and} \bibinfo{person}{Zhongzhi Zhang}.} \bibinfo{year}{2021}\natexlab{}.
\newblock \showarticletitle{Minimizing polarization and disagreement in social networks via link recommendation}.
\newblock \bibinfo{journal}{\emph{Advances in Neural Information Processing Systems}}  \bibinfo{volume}{34} (\bibinfo{year}{2021}), \bibinfo{pages}{2072--2084}.
\newblock


\end{thebibliography}

\section{Appendix} \label{appendix}

\subsection{Results for Alternative Definition}\label{Xue_result_appendix}

%\ahad{It seems in the main text we don't mention about this.\mohammad{red}{Yes, we mentioned it just a paragraph before "our contribution" part.}}

In this section, we explore the alternative definition of polarization proposed by Xu et al.~\cite{xu2022effects}. We believe that our proof techniques can be leveraged to produce similar results in this setup. Below, we show how some of our results can be easily expanded to this setting.

In the new setup, in the presence of stubborn users, polarization is defined as follows:
\begin{equation}\label{polarization_xu}
    \mathcal{P}_{G,z^*}:=\sum_{u \in V} k_u \bar{z}_u^2=\bar{\textbf{z}}^T \textbf{K} \bar{\textbf{z}}, 
\end{equation}
%\ahad{You are using both $u$ and $i$ in the sum!\mohammad{red}{done}}
where $\bar{z}$ is the mean-centralized version of $z$ and $k_u$ is the stubbornness of node $u \in V$. Considering disagreement to be the original definition in Equation~\eqref{disagreement}, the $\mathcal{PD}$ of social network $G$ with Laplacian matrix $\textbf{L}$, stubbornness matrix $\textbf{K}$ and innate opinions $ s $  can be simplified to
\begin{equation}\label{PD}
     \mathcal{PD}:=\mathcal{P}_{G,z^*}+ \mathcal{D}_{G,z^*}=\bar{s}_{\scriptscriptstyle\textbf{K}}^T \textbf{K} (\textbf{K}+\textbf{L})^{-1} \textbf{K} \bar{s}_{\scriptscriptstyle\textbf{K}}.
\end{equation}

In the following proposition, we will modify Theorem~\ref{P_PD_K_alpha_I_theorem} to address the case where the $\mathcal{PD}$ changes under the assumption of homogeneity, i.e., when $\textbf{K} = \alpha \textbf{I}$. Note that in this scenario, the matrix expression $\textbf{K} (\textbf{K} + \textbf{L})^{-1} \textbf{K}$ simplifies to $\alpha (\textbf{I} + \alpha^{-1} \textbf{L})$.

%\ahad{You use the word case 3 times in this sentence. And alse the sentence after. \mohammad{red}{lol! I fixed it}}

% In the following proposition, we will modify Theorem~\ref{P_PD_K_alpha_I_theorem} ( the case where the $\mathcal{PD}$ changes in the case of homogeneous case, i.e., $\textbf{K}=\alpha \textbf{I}$). Note that in this case, the matrix $\textbf{K} (\textbf{K}+\textbf{L})^{-1} \textbf{K}$ will be simplified to $\alpha (\textbf{I}+\alpha^{-1}\textbf{L})$. 

\begin{theorem}\label{P_PD_K_alpha_I_theorem_xu}
Consider a simple, undirected, connected graph \( G = (V, E) \) with \( |V| = n \) and given innate opinion vector $s \in \mathbb{R}^n$, where $\| s \| \leq R$. Then for any  $\alpha \leq \beta$, 
\begin{equation*}
   \mathcal{PD}(\alpha) \leq \mathcal{PD}(\beta), 
\end{equation*}
and 
\begin{equation*}
\mathcal{PD}(\beta)-\mathcal{PD}(\alpha) \leq (\beta-\alpha) R^2. 
\end{equation*}
\end{theorem}
\begin{proof}
The first part of the proof follows from the same argument as presented in Theorem~\ref{P_PD_K_alpha_I_theorem}. For the proof of the second part of the proof, our approach is similar to the Theorem~\ref{p_changes_upper_bound}. We define $E(x)$ as follows 
\begin{equation*}
E(x):=f(x,\beta)-f(x,\alpha), \quad f(x,c)=\frac{c}{1+c^{-1}x},
\end{equation*}
and note that $E'(x) <0$, which means $E(x) \leq E(0)=\beta-\alpha$. Thus, taking into account that $\| \bar{s} \|_2 \leq \| s \|_2 \leq R$, the proof is thereby concluded.
\end{proof}

In the case of stochastic block model, we also have the following theorem. 

\begin{theorem}\label{SBM_theorem_xu}
Let $G=(V, E)$ be the stochastic block model with $|V|=n$, where the expected adjacency matrix defined in equation (\ref{adjacency_A}), and the innate opinion given as Equation~\eqref{intenal_opinion_SBM}. If the stubbornness matrix is given by $\textbf{K}=\alpha \textbf{I}$ for any given real value $\alpha \geq 1$, then the $\mathcal{PD}$ index will be obtained as 
\begin{equation*}
    %\P(\alpha)=\frac{\alpha^3 n}{(nq+\alpha)^2}, \quad \quad 
    \mathcal{PD}(\alpha)=\frac{\alpha^2 n}{nq+\alpha}. 
\end{equation*}
\end{theorem}
\begin{proof}
The proof will be followed with the same argument as Theorem~\ref{SBM_theorem} with this different that the $\mathcal{PD}$ matrix is now given by $\alpha (\textbf{I}+\alpha^{-1}\textbf{L})^{-1}$. %\ahad{I think we should give a bit more explanations here! \mohammad{red}{I added a sentence, the total difference is less than 5 percent, the argument is the same}} 
\end{proof}

\end{document}